\documentclass[12pt, a4paper, oneside]{article}
\usepackage[utf8]{inputenc}
\usepackage{graphicx}

\usepackage{titlesec}
\graphicspath{{images/}}
\usepackage[a4paper,top=15mm,bottom=15mm,width=165mm,bindingoffset=6mm]{geometry}
\usepackage{amssymb}
\usepackage{nomencl}
\makenomenclature
\usepackage{caption}
\usepackage{subcaption}
\usepackage{epigraph}
\usepackage{amsthm}
\usepackage{pdfpages}
\usepackage{setspace}
\usepackage{nameref}
\usepackage{comment}
\usepackage{enumerate}  
\RequirePackage[breaklinks,colorlinks]{hyperref}

\usepackage{url}
\usepackage[parfill]{parskip}
\usepackage{stackengine}
\usepackage{xfrac}
\usepackage{mathtools}

\usepackage{braket}
\usepackage[nottoc]{tocbibind}

\usepackage[square,numbers]{natbib}
\usepackage{titlesec}

\newtheorem{myassumption}{\,\,\,\,\,\, \normalfont\it Assumption}

\newtheorem{definition}{Definition}
\newtheorem{notation}{Notation}

\newtheorem{remark}{Remark}
\newtheorem{lemma}{Lemma}
\newtheorem{proposition}{Proposition}

\hypersetup{colorlinks=true, linkcolor=black}

\newtheorem{theoremA}{\,\,\,\,\,\, \normalfont\scshape Theorem}

\newtheorem{theoremB}{\,\,\,\,\,\, \normalfont\scshape Theorem}

\newtheorem{theoremC}{\,\,\,\,\,\, \normalfont\scshape Theorem}

\definecolor{myblue}{rgb}{0,0,0.5}

\numberwithin{equation}{section}

\usepackage{etoolbox}
\makeatletter
\patchcmd{\endmyassumption}{\@endpefalse}{}{}{}
\makeatother

\usepackage{authblk}
\providecommand{\keywords}[1]{\textbf{\textit{Keywords:}} #1}

\date{August, 2024}
\title{Cycle-Star Motifs: Network Response to Link Modifications\thanks{This paper was published in Journal of Nonlinear Science, Volume 34, article number 60, (2024).}}

\author[1,2]{Sajjad Bakrani}

\author[1,2]{Narcicegi Kiran}
\author[1]{Deniz Eroglu}
\author[2,3]{Tiago Pereira}%\email{tiago.pereira@imperial.ac.uk}
\affil[1]{Faculty of Engineering and Natural Sciences, Kadir Has University, 34083, Istanbul, Turkey}

\affil[2]{Instituto de Ci\^encias Matem\'aticas e Computac\~ao, Universidade de S\~ao Paulo, São Carlos, 13566-590, Brazil}

\affil[3]{Department of Mathematics, Imperial College London, London SW7 2AZ, United Kingdom}
\begin{document}
\maketitle

\begin{abstract}
Understanding efficient modifications to improve network functionality is a fundamental problem of scientific and industrial interest. We study the response of network dynamics against link modifications on a weakly connected directed graph consisting of two strongly connected components: an undirected star and an undirected cycle. We assume that there are directed edges starting from the cycle and ending at the star (master-slave formalism). We modify the graph by adding directed edges of arbitrarily large weights starting from the star and ending at the cycle (opposite direction of the cutset). We provide criteria (based on the sizes of the star and cycle, the coupling structure, and the weights of cutset and modification edges) that determine how the modification affects the spectral gap of the Laplacian matrix. We apply our approach to understand the modifications that either enhance or hinder synchronization in networks of chaotic Lorenz systems as well as Roessler. Our results show that the hindrance of collective dynamics due to link additions is not atypical as previously anticipated by modification analysis and thus allows for better control of collective properties.

\end{abstract}
\newpage
\tableofcontents

\section{Introduction}

Many systems in nature are modelled as networks of interacting units with examples ranging from neuroscience \cite{ermentrout2010mathematical} to engineering  \cite{newman2018networks}. Recent work has revealed that the network interaction structure plays a crucial role in the network emergent dynamics \cite{Eroglu2017synchronisation, Prasad2010amplitude, Louodop2019extreme, Aguiar2019feedforward, Field2015heteroclinic}. Predicting the impact of the network structure on the dynamics is an intricate nonlinear problem that leads to many unexpected results. Indeed, in some situations improving the network structure may lead to functional failures such as Braess's paradox \cite{Eldan2017braess} and synchronization loss \cite{Pade2015improving, nishikawa2010network}. In large networks depending on the interaction function and isolated dynamics of the nodes, a topological hub may fail to be a functional hub \cite{TanziJEMS2020, eguiluz2005scale}.

%\textcolor{red}{
The effects of network topology on dynamical phenomena, such as synchronization, diffusion and random walks can be related to spectral properties of the graph, see for instance \cite{Chung1997spectral,Eroglu2017synchronisation}. Indeed, to predict the consequences of network modification on the dynamics, one needs to investigate the highly nonlinear changes in the spectrum of the graph Laplacian \cite{poignard2019effects, Pade2015improving,biyikoglu2007laplacian}

Although certain correlations between network structure and dynamics have been observed in experimental \cite{Hart2015adding} and theoretical \cite{Pade2015improving, Milanese2010approximating} investigations, most of these results are concerned with small modifications to the network. There is a lack of rigorous results to determine the relationship between the network structure and its dynamic properties for arbitrary size modifications. Most of the results in this direction rely on the modification theory of eigenvalues to determine which structural changes are detrimental to the network dynamics. However, previous results relying on perturbation theory suggest that desynchronizing the network by adding new links is unusual \cite{poignard2019effects}. To understand this problem, we need to unveil the full nonlinear picture and deal with large changes in the topology.

Networks are a combination of motifs that dictate dynamical behavior and provide resilience to the overall system~\cite{ma2008ordered, kashtan2005spontaneous}. We focus on two motifs of complex networks -- a cycle and a star -- since they are the main constituents of important networks. 
Indeed, cycles are typical components in the nervous system \cite{alexander1986} and orientation tuning in visual cortex \cite{ben1997traveling}. Also, in the context of neuroscience highly connected nodes, called hubs, play a fundamental role in the network \cite{bonifazi2009gabaergic}. These networks with hubs are modeled as a collection of star motifs, and each star motif is capable of generating intricate dynamics \cite{vlasov2015explosive}, as well as their overall interaction \cite{tonjes2021coherence}.

Although both cycle and star motifs were investigated for noteworthy network dynamics such as collective behavior \cite{mersing2021novel,muni2020chimera,kantner2013bifurcation,manik2017cycle,corder2023emergence}  and both motifs have a fully developed spectral theory \cite{Brouwer2011spectra}  their eigenvectors and eigenvalues can be fully described (as in the case of rings where the matrices are circulant), when these motifs are coupled,  the eigenvalues problem becomes an intricate nonlinear problem that remains open.

In this paper, we consider models of networks consisting of cycles and stars coupled in a master-slave topology. Although our problem is dynamics-motivated, we state our main results in a graph theoretic form and consider the synchronization as an application. This is because, in a broader sense, the spectral properties of the graph Laplacian are important in the study of graph connectedness and, hence, any phenomena related to this concept \cite{mohar1997some}.

%\textcolor{red}{In the application of our theorems to network synchronization, we used the results based on the theory developed in [PERV14]. When the coupling function is nonlinear or is a projector we can rely on the master stability function formalism. (I added this here, maybe we should relocate it.)}

%\textcolor{red}{I'm citing two papers here for now: \cite{Bonifazi2009gabaergic} and \cite{Popovych2011delay}.}

\subsection{Informal statements of our results}

We consider three models illustrated in Figures \ref{Fig8767rt6r5jytlo7ro6se}, \ref{Figi867ssvsutie7r57trlr8trwlegkfyuga}, and \ref{Fig9798v7tc5extdht4ebrdbgrstgds}. All these three models have a master-slave structure, a cycle \(C_{n}\), a star \(S_m\), and cutset edge(s) starting from the cycle and ending at the hub of the star. We modify these networks and break the master-slave structure by adding directed links from the star to the cycle (red-color edges in the figures).

\begin{figure}[h]
\begin{subfigure}[]{0.5\textwidth}
\center
\includegraphics[scale=0.12]{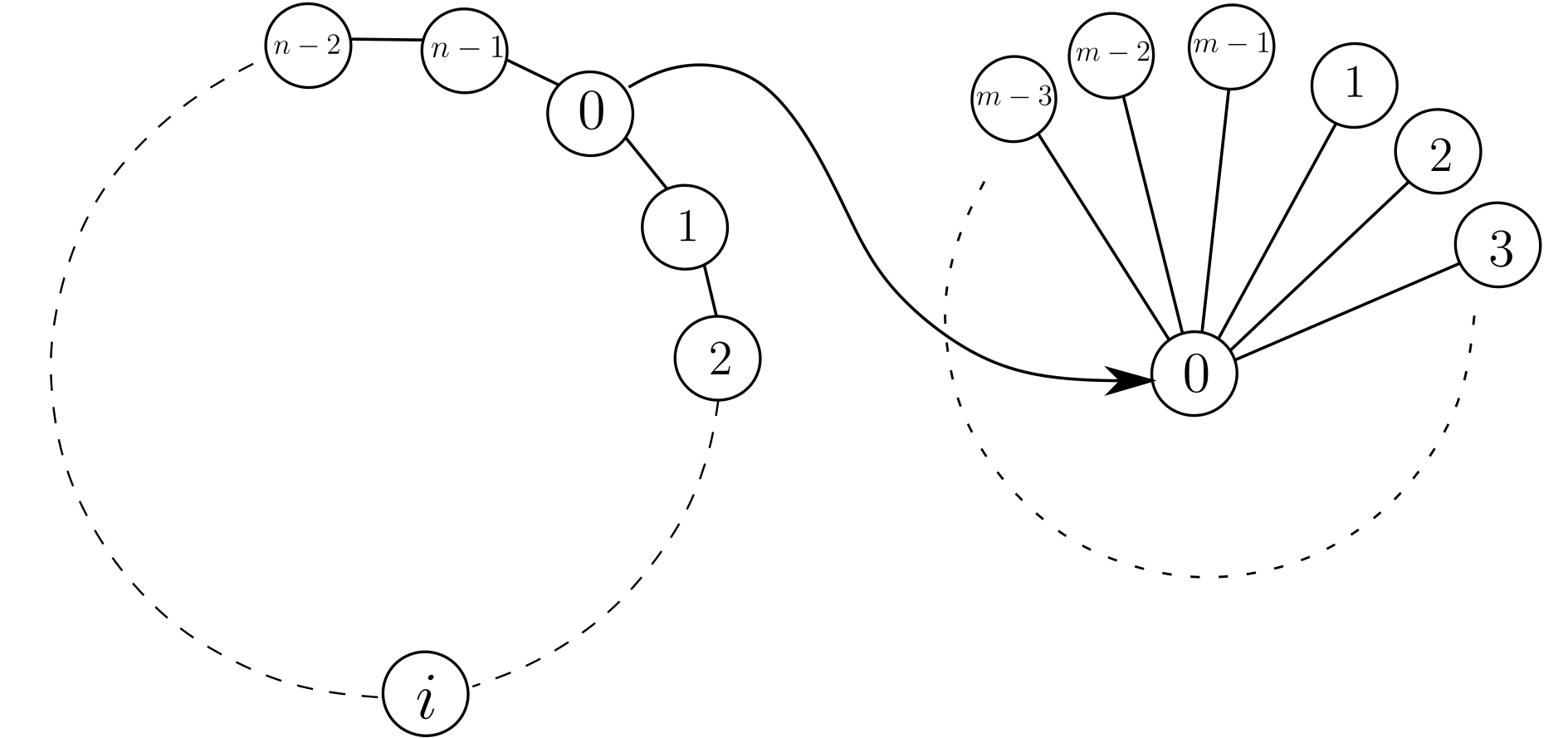}
\caption{The unmodified graph \(G\).}
\label{Fig345uru5rtjyrdtrxgdxw}
\end{subfigure}
%\vspace{5mm}
%\newline
\begin{subfigure}[]{0.5\textwidth}
\center
\includegraphics[scale=0.12]{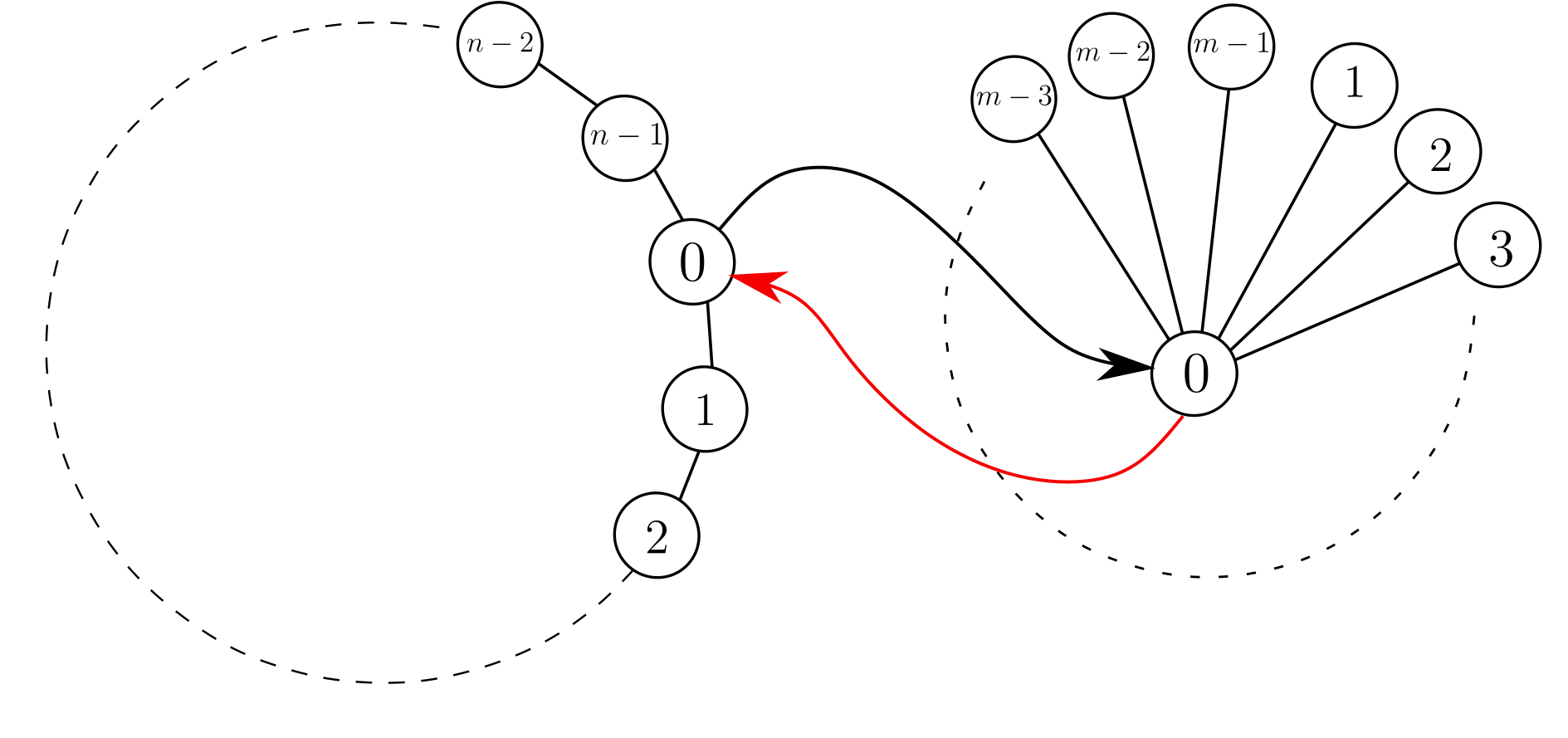}
\caption{The modified graph \(G_{p}\).}
\label{Fig54kugguiyou038byaas}
\end{subfigure} 
\caption{Model I: Breaking the master-slave through hub coupling. We add a directed link from the hub of the star to the cutset node (the red-color edge) where the cutset node refers to the node which the cutset edge starts from and the weakly connected graph becomes strongly connected. The weight of each of the black-color edges is one while the weight of the red-color edge (modification edge) is arbitrary.}
\label{Fig8767rt6r5jytlo7ro6se}
\end{figure}
\begin{figure}[h]
\begin{subfigure}[]{0.5\textwidth}
\center
\includegraphics[scale=0.11]{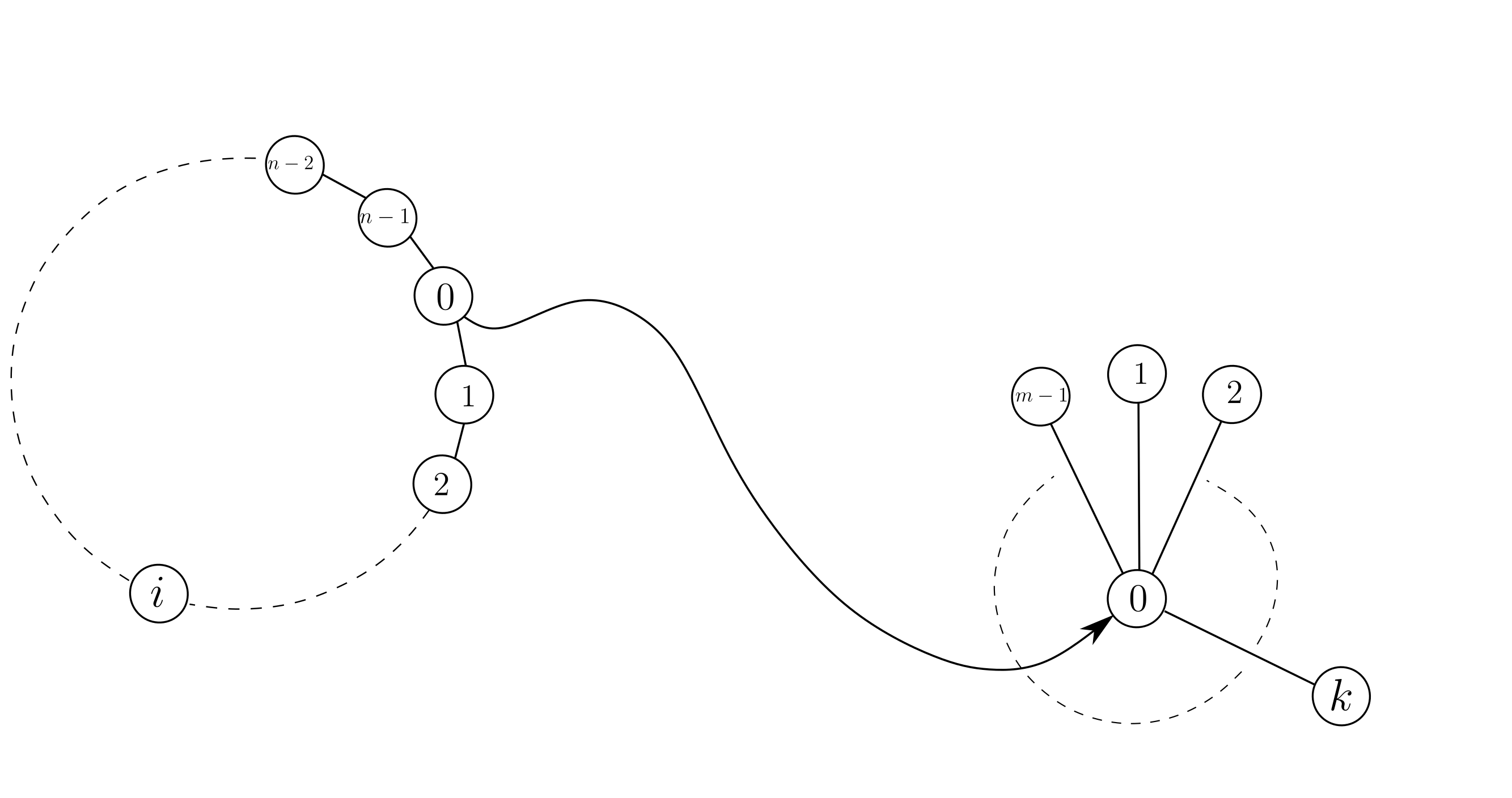}
\caption{The unmodified graph \(G\).}
\label{Fig97o8gy7tuc7itca6rekrkwuayes}
\end{subfigure}
%\vspace{5mm}
%\newline
\begin{subfigure}[]{0.5\textwidth}
\center
\includegraphics[scale=0.11]{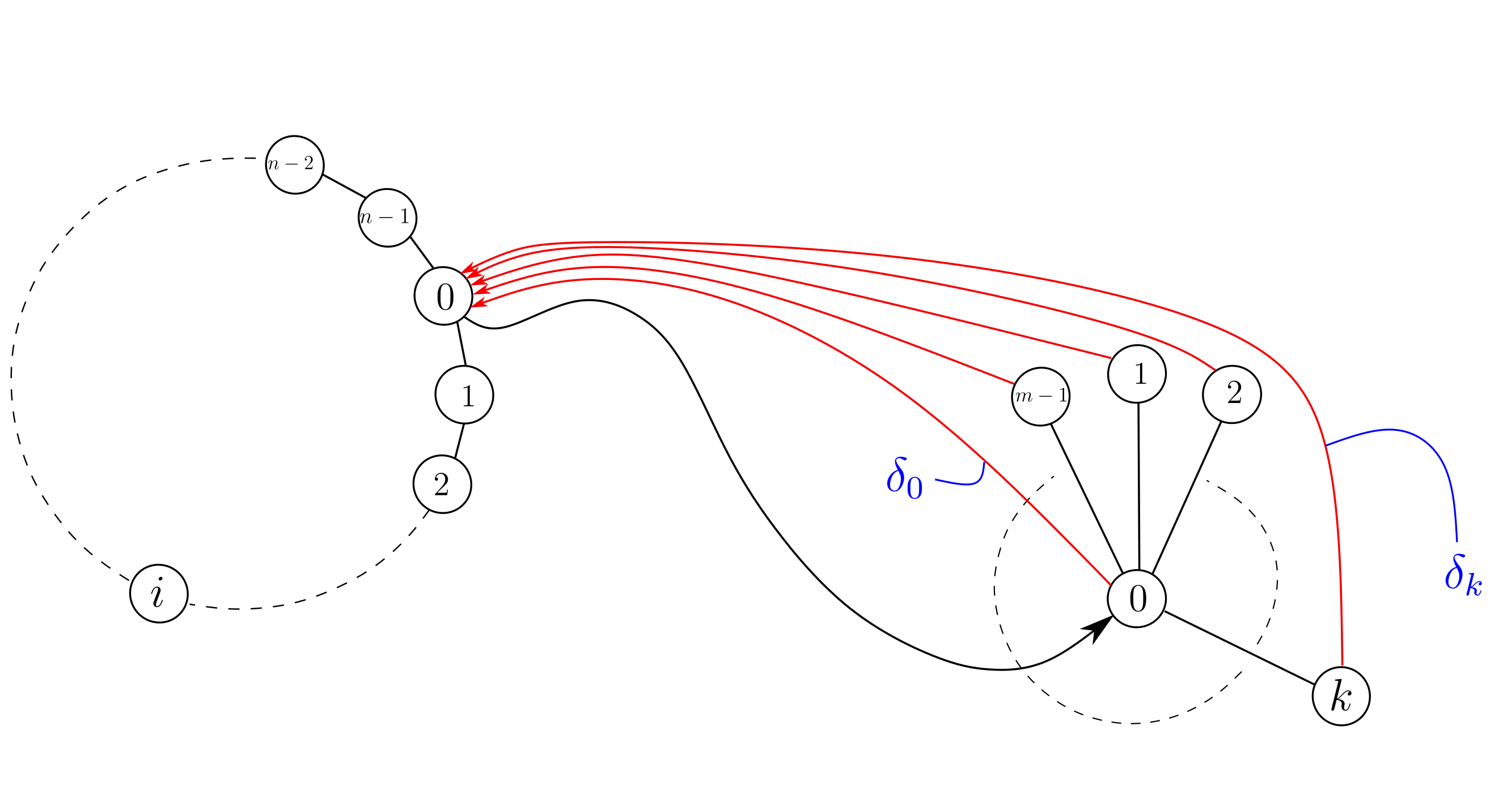}
\caption{The modified graph \(G_{p}\).}
\label{Figkiut765j6ra8rekyefk7wrkweewdfsf}
\end{subfigure}
\caption{Model II: Breaking the master-slave through multiple couplings. We add links from some nodes of the star to the cutset node of the cycle. The weight of each of the black-color edges is one while the weights of the red-color edges (modification edges) are arbitrary.}
\label{Figi867ssvsutie7r57trlr8trwlegkfyuga}
\end{figure}

\begin{figure}[h!]
\begin{subfigure}[]{0.5\textwidth}
\center
\includegraphics[scale=0.11]{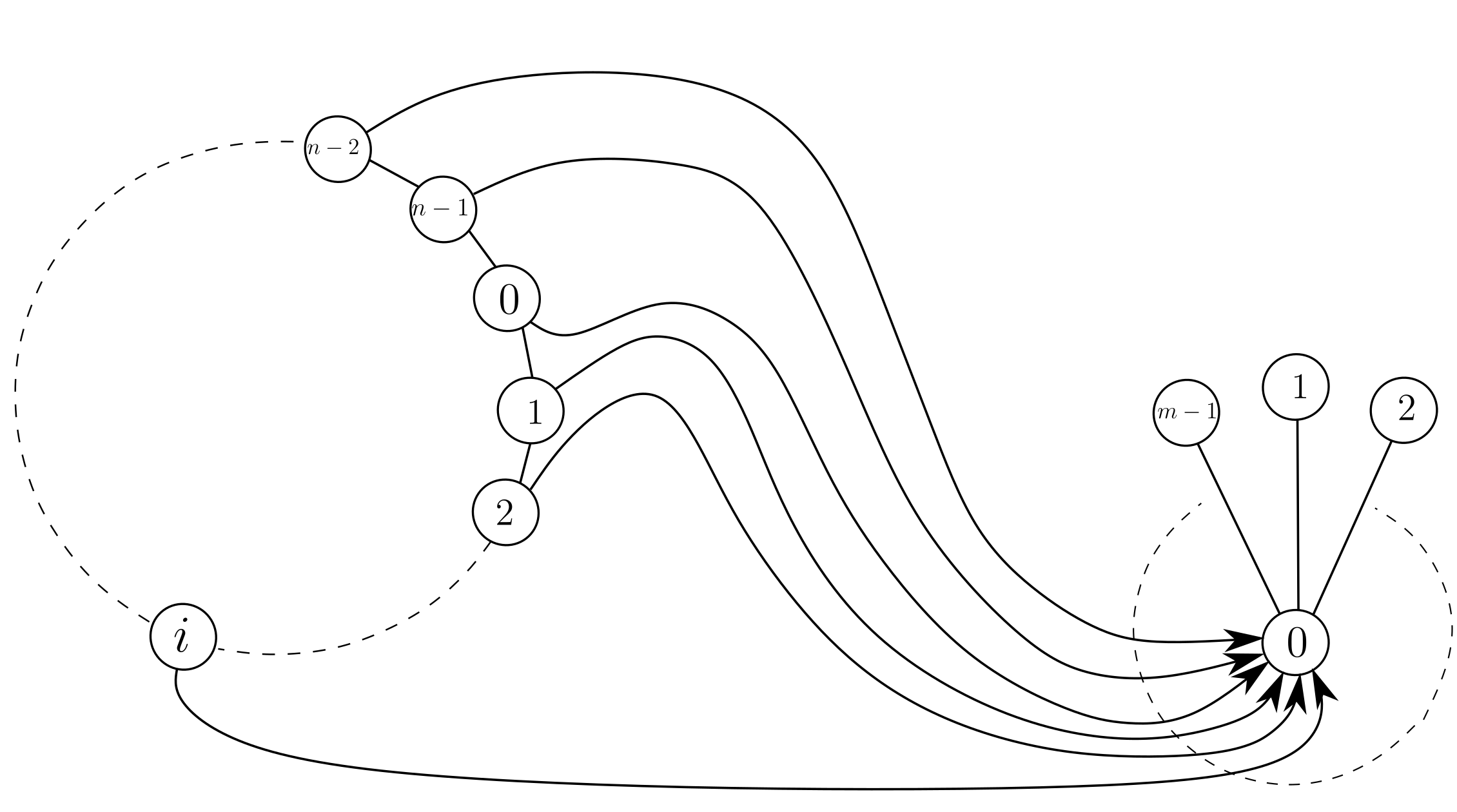}
\caption{The unmodified graph \(G\).}
\label{Fig9798v7tc5extdht4ebrdbgrstgds111}
\end{subfigure}
%\vspace{5mm}
%\newline
\begin{subfigure}[]{0.5\textwidth}
\center
\includegraphics[scale=0.11]{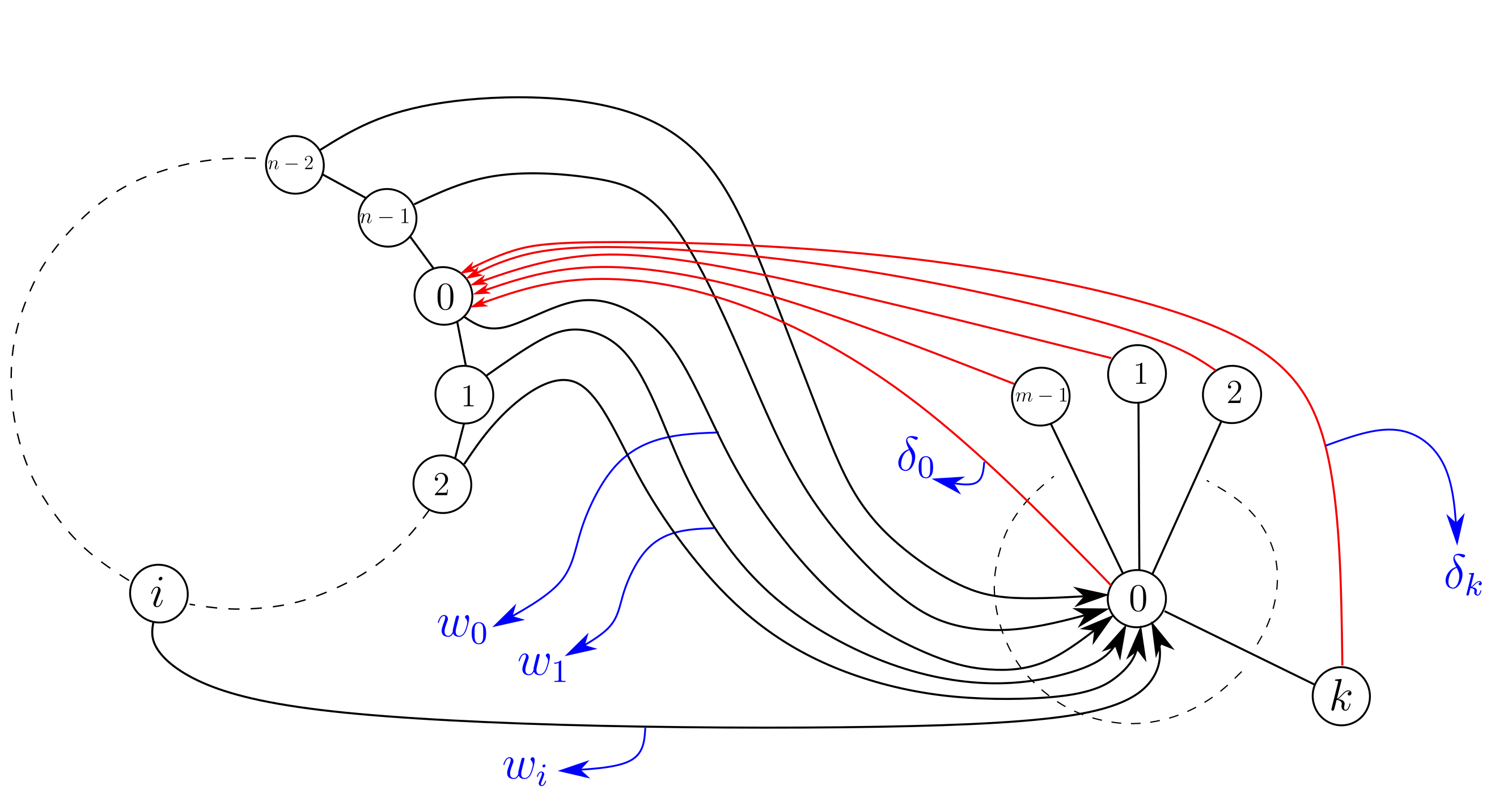}
\caption{The modified graph \(G_p\).}
\label{Fig9798v7tc5extdht4ebrdbgrstgds222}
\end{subfigure}
\caption{Model III: Breaking the generalized master-slave through multiple couplings. We add links from nodes of the star to the one cutset node of the cycle. The weight of each of the black-color edges is one while the weights of the red-color edges (modification edges) are arbitrary.}
\label{Fig9798v7tc5extdht4ebrdbgrstgds}
\end{figure}

\subsubsection{Adjacency matrices and graph Laplacians}

Let \(G\) be a weighted directed graph (digraph) whose nodes are labelled by \(1,\ldots, n\). We define the adjacency matrix of \(G\) by \(A_{G} = (A_{ij})\), where \(A_{ij}\geq 0\) is the weight of the directed edge starting from node \(j\) and ending at node \(i\). The in-degree of a node is the sum of the weights of the edges that the node receives from other nodes, i.e., the in-degree of the node $i$ is $\sum_{j} A_{ij}$. We define the Laplacian matrix of \(G\) by \(L_{G} := D_{G} - A_{G}\), where \(D_G\) is a diagonal matrix whose \((i,i)\)-entry is the in-degrees of the node \(i\) of \(G\). Let \(L_{G}\) and \(L_{G_p}\) represent the Laplacians of the unmodified and modified graphs, respectively. Let \(\lambda_{2}(L_{G})\) and \(\lambda_{2}(L_{G_{p}})\) be the associated second minimum eigenvalues, so-called spectral gap. Our results explain how the modification affects the spectral gap of the Laplacian matrices of these models. We provide more details in Section \ref{LaplacianMatrices}.

\subsubsection{Results (informal version)}

Assume \(\delta_{0}\geq 0\) is the weight of the modification edge starting from the hub and \(\delta\geq 0\) is the sum of the weights of all the modification edges. In model I, we have \(\delta_{0} = \delta\), and in the other two models, \(\delta_{0} \leq \delta\). Let \(m\) and \(n\) be the sizes of the star and cycle, respectively. The term \(o(1)\) in the informal statements of the following Theorems \ref{Thm8oyrq8yr4q} and \ref{Thm8y8v7rti473rt} (resp. Theorem \ref{Thm233yisu3y3rqu}) stands for a function of \(m\) (resp. \((m,w)\)) that converges to \(0\) as \(m\rightarrow \infty\) (resp. \(\frac{m}{w}\rightarrow \infty\)).
When the weight of the modification is small, we will call this modification local. This is because the results follow from local analysis of the eigenvalues. If the weight of the modification is large we called it global, as the analysis require global techniques to gain insights on the eigenvalues. 
These models are discussed precisely in Section \ref{Section87y7gq3gid}. Here, we give an informal version of our main results.

\begin{theoremA}[Informal statement]
Consider model I illustrated in Figure \ref{Fig8767rt6r5jytlo7ro6se}. Let the modification \(\delta > 0\) be arbitrary (it does not need to be sufficiently small). We have
\begin{enumerate}
\item Although \(L_{G_{p}}\) is not necessarily symmetric, all of its eigenvalues are real.
\item There exists a critical cycle size \(n_{c} = \pi\sqrt{m+1}[1 + o(1)]\) such that \(\lambda_{2}\left(L_{G_{p}}\right) < \lambda_{2}\left(L_{G}\right)\) if and only if \(n \geq n_{c}\).
\end{enumerate}
\label{Thm8oyrq8yr4q}
\end{theoremA}

We illustrate Theorem A in Figure \ref{curveintensity}.

\begin{figure}[h!]
\centering
\includegraphics[scale=0.58]{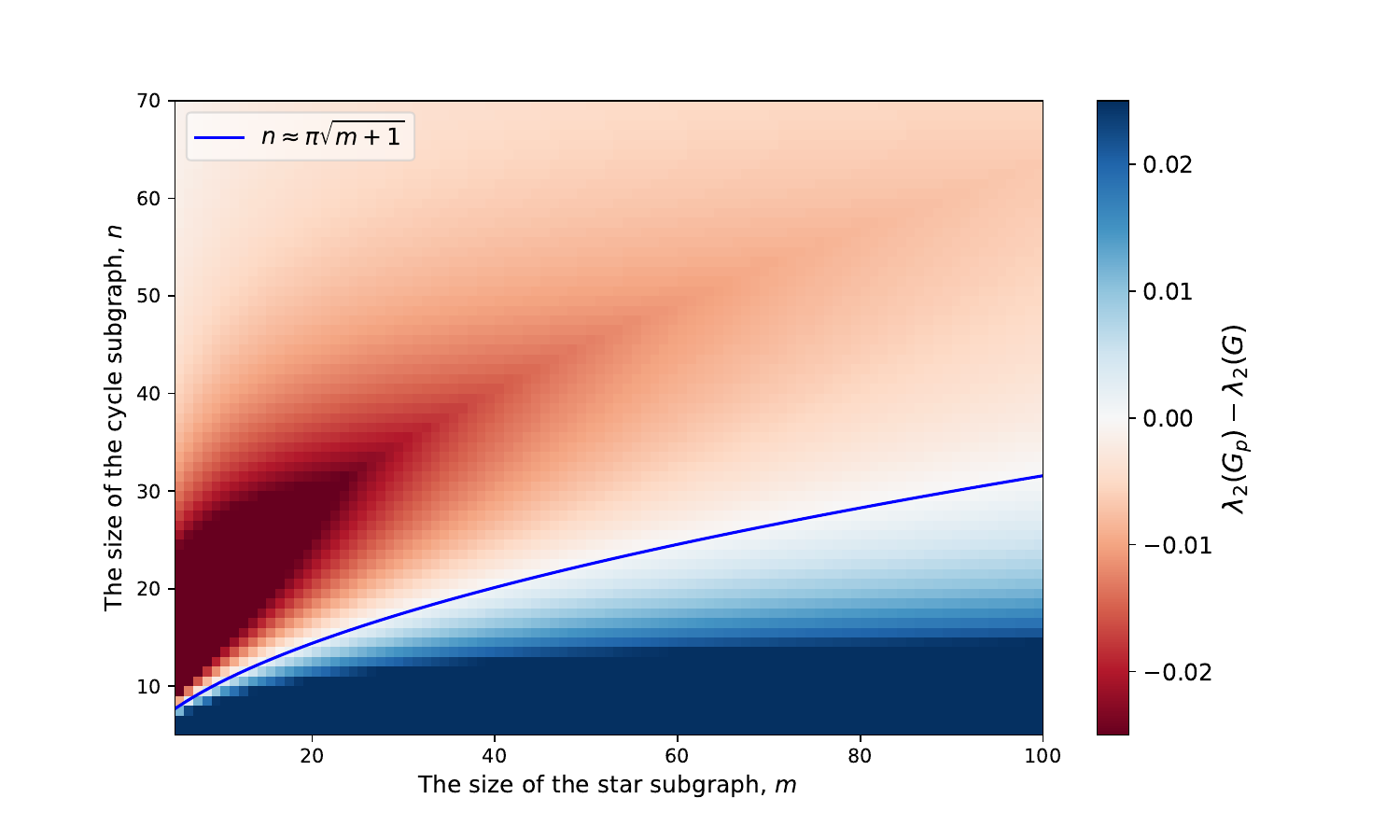}
\caption{A comparison of the cases in Theorem A and computations of $\lambda_2(G_p)-\lambda_2(G)$. We create a graph $G$ as described in model \ref{Fig345uru5rtjyrdtrxgdxw} with cycle $C_n$ and star $S_m$ subgraphs whose sizes are $n$ and $m$, respectively. Then we modify the network as shown in model \ref{Fig54kugguiyou038byaas} where $\delta_0=1$ and calculate the difference $\lambda_2(G_p)-\lambda_2(G)$ to characterize the behavior of the second minimum eigenvalue after modification. The red color of grids corresponds to a decrease of the second minimum eigenvalue after modification, and the blue color of grids corresponds to an increase of the second minimum eigenvalue after modification, where the intensity of the color at each grid shows the size of the difference $\lambda_2(G_p)-\lambda_2(G)$. Simultaneously, the blue curve given by $n_c=\pi \sqrt{m+1}$ is shown. Thus, regions separated by the blue curve manifest the signature of $\lambda_2(G_p)-\lambda_2(G)$ and it corresponds to the critical transition between the cases stated in Theorem A, i.e., decreasing or increasing behavior of the second minimum eigenvalue after modification.}
\label{curveintensity}
\end{figure}

To give the informal statement of Theorem \ref{Theorem18na75t87ckq7tyrm3wa}, let \(\rho := \frac{\delta_{0}}{\delta}\). This ratio can be seen as a measure of the modification that the cycle receives from the hub of the star relative to the modification it receives from the leaves of the star. We have \(\rho \leq 1\), and by setting \(\rho = 1\), the model II reduces to model I. 

\begin{theoremB}[Informal statement]
Consider model II illustrated in Figure \ref{Figi867ssvsutie7r57trlr8trwlegkfyuga}. We have
\begin{enumerate}
\item Under a local modification, the statement of Theorem \ref{Thm8oyrq8yr4q} is valid for model II. When \(\delta > 0\) is sufficiently small, all the eigenvalues of \(L_{G_p}\) are real and, the modification decreases the spectral gap if and only if the size of the cycle is larger than the critical value \(n_{c}(m)\).
\item Under a global modification (\(\delta\) be arbitrary), the statement of Theorem \ref{Thm8oyrq8yr4q} is valid for model II when \(\rho > K\), where \(0 < K < 1\) is a constant given in Section \ref{Section87y7gq3gid}.
\item Under a global modification (\(\delta\) be arbitrary), we have \(\mathrm{Re}\left(\lambda_{2}\left(L_{G_{p}}\right)\right) < \lambda_{2}\left(L_{G}\right)\) if the size of the cycle is larger than a critical value \(n^{*}_{c}(m) = 2 \pi\sqrt{m+1}[1 + o(1)]\).
\end{enumerate}
\label{Thm8y8v7rti473rt}
\end{theoremB}

\begin{theoremC}[Informal statement]
Consider model III illustrated in Figure \ref{Fig9798v7tc5extdht4ebrdbgrstgds}. Let \(w\) be the sum of the weights of all cutset edges. There exist two critical \(n_{c}(m,w)\) and \(n^{*}_{c}(m,w)\) such that \(n_{c}^{*}\approx 2 n_{c} = 2 \pi\sqrt{\frac{m}{w} + 1}[1 + o(1)]\), and
\begin{enumerate}
\item Under a local modification, we have
\begin{enumerate}
\item if \(n > n^{*}_{c}(m,w)\), then \(\lambda_{2}(L_{G_p})\leq \lambda_{2}(L_G)\).
\item if \(n_{c}(m,w) < n < n^{*}_{c}(m,w)\), then both increasing and decreasing in the spectral gap can happen. See Theorem \ref{Thm274i7tigfdyf} for the distinction between cases.
\item if \(n < n_{c}(m,w)\), then \(\lambda_{2}(L_{G_p}) > \lambda_{2}(L_G)\).
\end{enumerate}
\item Under a global modification, if \(n > n^{*}_{c}(m,w)\), then \(\mathrm{Re}(\lambda_{2}(L_{G_p}))\leq \lambda_{2}(L_G)\).
\end{enumerate}
\label{Thm233yisu3y3rqu}
\end{theoremC}

In all three mentioned models, we consider the scenario in which the cutset edges start from the cycle and end at the hub of the star, briefly called the hub connection. Another scenario that can be considered is where the cutset edges end at the leaves of the star instead of its hub, briefly called leaf connection. Our numerical investigation shows that there exists a critical \(n_c'\) analogous to $n_c$ in Theorem A and a critical $n'^{*}_c(m)$ analogous to $n^{*}_{c}(m)$ in Theorem B for the leaf connection. However, $n_c'$ is bounded below by $n_c$, likewise $n'^{*}_c(m)$ is bounded below by $n^{*}_c(m)$. In other words, if we compare the incidence number of \(\lambda_{2}\left(L_{G_{p}}\right) < \lambda_{2}\left(L_{G}\right)\) for hub and leaf connection, hub connection maximizes the incidence number of \(\lambda_{2}\left(L_{G_{p}}\right) < \lambda_{2}\left(L_{G}\right)\) for the same parameter set.

The Laplacian matrix \(L_{G}\) of the unmodified graph in all the models I, II, and III is a block lower-triangular matrix, see the form (\ref{triangularstructure}). However, adding a modification in the opposite direction of the cutset breaks the triangular structure of \(L_{G}\), which turns the analysis of its spectral gap into a non-trivial problem. Our approach to analyzing the changes in the spectral gap consequent to the graph modification is to investigate a secular equation of the Laplacian matrix and its roots. Our analysis is not restricted to the local modification\footnote{note that even for investigating local changes in the spectral gap, the standard approach, see e.g. Theorem 6.3.12 in \cite{horn2012matrix}, cannot be applied since the spectral gaps in our models are not necessarily simple}; we indeed analyze the change in the spectral gap under modification of arbitrary size. This requires further work on not only analyzing the modification of spectral gap but also understanding the modifications and distribution of the whole spectrum of the Laplacian matrix.

\section{Applications to synchronization}

We consider synchronization in networks of diffusively coupled oscillators as an application. Consider a triplet \(\mathcal{G} = (G, f, H)\), where \(G\) is a weighted digraph, and \(f, H\in \mathcal{C}^{1}(\mathbb{R}^{l})\) for \(l\geq 1\). The triplet \(\mathcal{G}\) defines a system of the form
\begin{equation}\label{eq87tvk7vksyueeagt5to}
\dot{x}_{i} = f\left(x_{i}\right) + \Theta\sum_{j=1}^{N} A_{ij} H (x_{j} - x_{i}),\qquad i = 1, 2,\ldots, N,
\end{equation}
where \(\Theta\geq 0\) is called the coupling strength. Each variable \(x_{i}\) represents the state of the $i$th node of \(G\), the function \(f\) describes the isolated dynamics at each node, and the function \(H\) is called the coupling function. We call the triplet \(\mathcal{G}\) or its associated system (\ref{eq87tvk7vksyueeagt5to}) a network of diffusively coupled (identical) systems.

We define the synchronization manifold as
\begin{equation}\label{eq23kugkerferiwqrhuwehg}
M := \left\lbrace \left(x_{1}, \ldots, x_{N}\right): \quad x_1 = \cdots = x_N \in U\right\rbrace.
\end{equation}
We say that a network \(\mathcal{G}\) synchronizes if there exists an open neighborhood \(V\) of \(M\) such that the forward orbit of any point in \(V\) converges to \(M\). It is shown \cite{Pereira2014towards} that for a network \(\mathcal{G}\) with a coupling strength \(\Theta\), if
\begin{enumerate}
\item the graph \(G\) has a spanning diverging tree, and
\item there exists an inflowing open ball \(U\subset \mathbb{R}^{l}\) which is invariant with respect to the flow of the isolated system \(\dot{x} = f(x)\), and we have \(\| Df(x)\| \leq K\) for some \(K > 0\) and all \(x \in U\), and
\item we have \(H(0) = 0\); moreover, all the eigenvalues of \(DH(0)\) are real and positive,
\end{enumerate}
then there exists \(\Theta_{c} \geq 0\) such that when \(\Theta \geq \Theta_{c}\), \(\mathcal{G}\) synchronizes. We call
\begin{equation}\label{eq08o8hyh7eqot43er}
\Theta_{c} = \frac{\rho}{ \mathrm{Re}(\lambda_{2})},
\end{equation}
the critical coupling strength where \(\rho = \rho(f, DH(0))\) is a constant. Note that if the third assumption is not fulfilled, the synchronization condition (\ref{eq08o8hyh7eqot43er}) may no longer be valid. However, in this case, new synchronization conditions may be obtained under the framework of master stability function formalism \cite{Eroglu2017synchronisation}. Relation (\ref{eq08o8hyh7eqot43er}) with assumptions stated above gives us a criterion to compare synchronizability in networks. More precisely,
\begin{definition}
Consider two networks \(\mathcal{G}_{1} = (G_{1}, f_{1}, H_{1})\) and \(\mathcal{G}_{2} = (G_{2}, f_{2}, H_{2})\) that satisfy the assumptions above. Let \(\Theta_{c}(\mathcal{G}_{1})\) and \(\Theta_{c}(\mathcal{G}_{2})\) be the critical coupling strengths of \(\mathcal{G}_{1}\) and \(\mathcal{G}_{2}\), respectively. We say that \(\mathcal{G}_{1}\) is more synchronizable than \(\mathcal{G}_{2}\) if \(\Theta_{c}(\mathcal{G}_{1}) < \Theta_{c}(\mathcal{G}_{2})\).
\end{definition}
Having \(\Theta_{c}(\mathcal{G}_{1}) < \Theta_{c}(\mathcal{G}_{2})\)  means that \(\mathcal{G}_{1}\) synchronizes for a larger range of \(\Theta\) than \(\mathcal{G}_{2}\). Let us now consider the case that two networks \(\mathcal{G}_{1}\) and \(\mathcal{G}_{2}\) only differ in their topology, i.e., having the same isolated dynamics and coupling functions while the graph structures can be different. In this case, following (\ref{eq08o8hyh7eqot43er}), the spectral gaps of the underlying graphs of the networks determine which one is more synchronizable. 

Let consider two networks \(\mathcal{G}_{1} = (G_{1}, f, H)\) and \(\mathcal{G}_{2} = (G_{2}, f, H)\) that satisfy the assumptions above. Moreover, let \(\lambda_{2}(G_{1})\) and \(\lambda_{2}(G_{2})\) be the spectral gaps of \(G_{1}\) and \(G_{2}\), respectively. The network \(\mathcal{G}_{1}\) is more synchronizable than \(\mathcal{G}_{2}\) if and only if \(\lambda_{2}(G_{1}) > \lambda_{2}(G_{2})\). 

\subsection{Synchronization of coupled Lorenz oscillators} \label{cycle-star-numerics}
We consider the following settings for model II given in Figure \ref{Figi867ssvsutie7r57trlr8trwlegkfyuga}:
Two networks \(\mathcal{G} = (G, f, H)\) and \(\mathcal{G}_{p} = (G_{p}, f, H)\) are generated where $G$ and $G_p$ are the unmodified and the modified graphs, respectively. The chosen isolated dynamics $f$ is the Lorenz oscillator
\begin{equation}
\begin{array}{ll}
\dot{x}=\sigma(y-x),\\
\dot{y}=x(\gamma-z)-y,\\
\dot{z}=xy-\beta z,
\end{array}
\end{equation}
where $\sigma=10$, $\gamma=28$, $\beta=8/3$. Here, \(H\) is the identity function on \(\mathbb{R}^3\). For the described setting, equation (\ref{eq08o8hyh7eqot43er}) can be written as $\Theta_c=\frac{\kappa}{\mathrm{Re}(\lambda_2)}$, where $\kappa$ is defined as in Section 5 of \cite{Eroglu2017synchronisation}.  We numerically find that $\kappa\approx0.9$. So, the expected values of $\Theta_c(\mathcal{G})$ and $\Theta_c(\mathcal{G}_p)$ are calculated accordingly. We examine two experiments to reveal how link addition can lead to synchronization in the network $\mathcal{G}_p$ or break the synchronization in the initial network $\mathcal{G}$ (see Figures \ref{sync111} and \ref{sync112}). The network of coupled Lorenz  oscillators in model II is simulated to show the synchronization error
\begin{equation}
\braket{E} = \sum_{i\neq j}\frac{\|x_i - x_j\|}{(n+m)(n+m-1)}.
\end{equation}

%\textcolor{red}{
It is worth mentioning that the same simulations that we have done for the Lorenz system can be done for other systems as well. Indeed, similar results hold as long as the initial conditions lead to an attractor in the synchronization manifold that is contained in a compact set. 
%(Before mentioning the Rossler, maybe we should mention some attractors that actually contained in a compact neighborhood, thenn we mention Rossler as an example that doesn't satisfy our condition. ?????????????) The R\"{o}ssler attractor for instance, depending on the initial conditions, the trajectories may diverge. Keeping this in mind, similar figures as Figure \ref{sync111} can be produced for other isolated dynamics and also models II and III.}

\subsection{Hindering synchronization}

To examine the hindrance of synchronization due to link addition, the overall coupling constant $\Theta$ is selected such that $\Theta_c(\mathcal{G})<\Theta<\Theta_c(\mathcal{G}_p)$ (see Figure \ref{sync111}). Note that such $\Theta$ values only exist when \(\lambda_{2}(G) > \lambda_{2}(G_{p})\) due to the order relations of synchronizability stated above. When the selected $\Theta$ is above the $\Theta_c(\mathcal{G})$, the trajectories synchronize for the network $\mathcal{G}$. Then, the system is modified by adding links at a given time $t$. Since the selected $\Theta$ is below the $\Theta_c(\mathcal{G}_p)$, the system loses its synchronization thereafter.

In model II, the sizes of the cycle and star subgraphs are set to $n=15$ and $m = 15$.  The weights of the cutset and modification edges are $w_0=1$ and $\delta_i=1$, where $i= 0, 1,\ldots, m-1 $. All initial states are randomly selected from the uniform distribution over $[3.5, 5)$.

\begin{figure}[h!]
\centering
\includegraphics[scale=0.5]{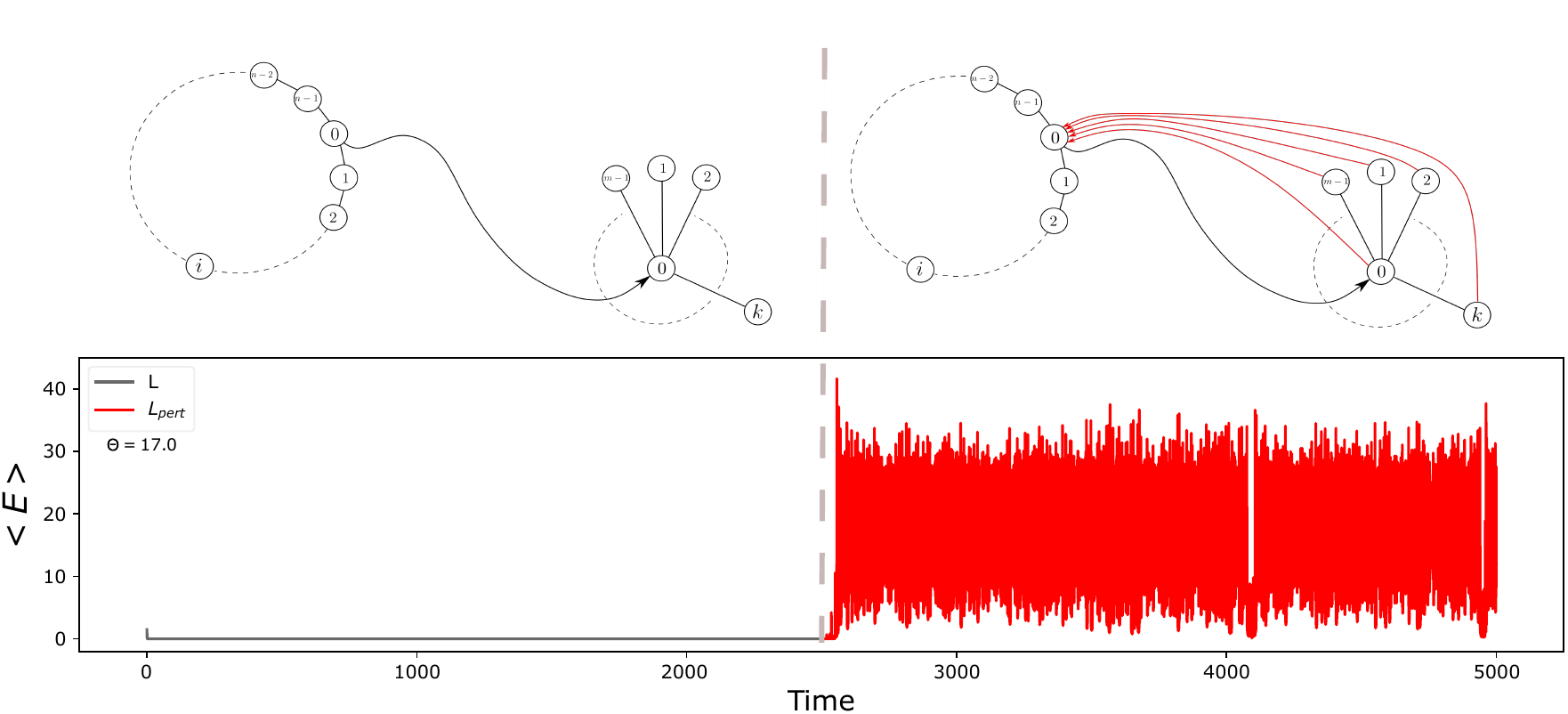}
\caption{Hindrance of synchronization due to link addition: Networks of coupled Lorenz oscillators in model II are simulated to show the synchronization error. The sizes of the cycle and star subgraphs are set to $n = 15$ and $m = 15$, and subgraphs are connected via a directed link from the cycle subgraph to the star subgraph where \(w_0 = 1\). We consider $H=\mathbf{I}$ as the coupling function. We choose initial conditions randomly selected from the uniform distribution over \([3.5, 5)\), and integrate the network until time $t=2500s$. The system goes into synchronization after some transient. At $t=2500s$, we add the red links to the system,  i.e., \(\delta_i = 1\), where \(i = 0,1,\hdots,m - 1\), and perturb the system by adding noise randomly selected from the uniform distribution over \([0.01, 0.02)\) to each state, then the synchronization loss occurs and the system doesn't return into synchronization after transient time. Note that, \(\alpha_1=0.04, \beta_{15,1}^-=0.06\) in Theorem B.}
\label{sync111}
\end{figure}

\subsection{Enhancing synchronization}
To examine the enhancement of synchronization due to link addition, the overall coupling constant $\Theta$ is selected such that $\Theta_c(\mathcal{G}_p)<\Theta<\Theta_c(\mathcal{G})$ (see Figure \ref{sync112}). Note that such $\Theta$ values only exist when \(\lambda_{2}(G_p) > \lambda_{2}(G)\). When the selected $\Theta$ is below the $\Theta_c(\mathcal{G})$, the trajectories cannot synchronize for the network $\mathcal{G}$. Then, the system is modified by adding links at a given time $t$. Since the selected $\Theta$ is above the $\Theta_c(\mathcal{G}_p)$, the system synchronizes.

\begin{figure}[h!]
\centering
\includegraphics[scale=0.5]{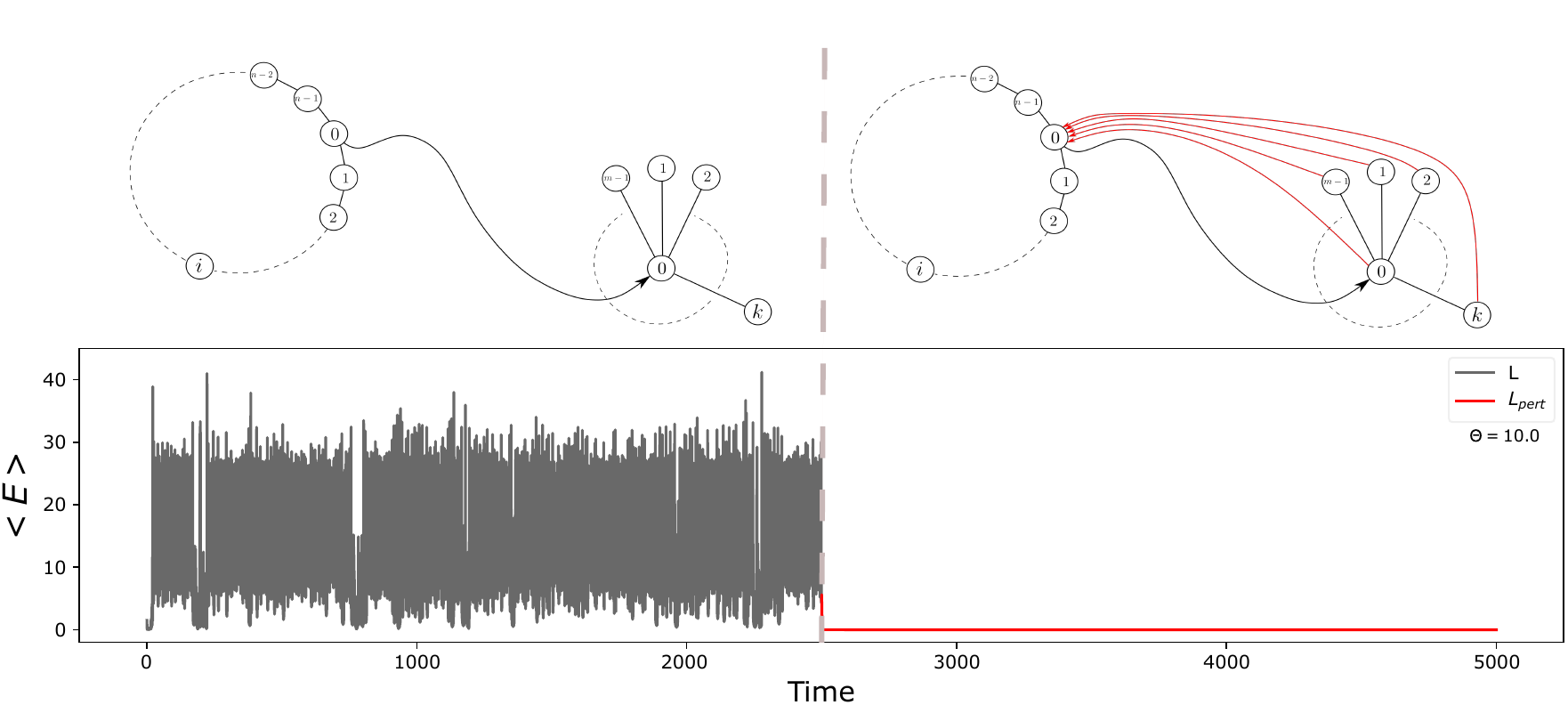}
\caption{Enhance of synchronization due to link addition: Networks of coupled Lorenz oscillators in model II are simulated to show the synchronization error. The sizes of the cycle and star subgraphs are set to $n = 9$ and $m = 15$, and subgraphs are connected via a directed link from the cycle subgraph to the star subgraph where \(w_0 = 1\). We consider $H=\mathbf{I}$ as the coupling function. We choose initial conditions randomly selected from the uniform distribution over \([3.5, 5)\), and integrate the network until time $t=2500s$.  After the red links are added to the system at $t=2500s$,  i.e., \(\delta_i = 1\), where \(i = 0,1,\hdots,m - 1\), the synchronization occurs where the mean error $\langle E\rangle$ goes to zero.  Note that, \(\alpha_1=0.12,\) \(\beta_{15,1}^-=0.06\) in Theorem B.}
\label{sync112}
\end{figure}

In model II, the sizes of the cycle and star subgraphs are set to $n=9$ and $m = 15$. The weights of the cutset and modification edges are $w_0=1$ and $\delta_i=1$, where $i= 0,1,\ldots,m-1 $. 
All initial states are randomly selected from the uniform distribution over $[3.5,5)$. Therefore, hindrance and enhancement of synchronization due to link addition manifest themselves in simulations, and it perfectly agrees with the findings of our theorems.

\section{Problem setting and results}\label{Section87y7gq3gid}

Let \(G\) be a weighted directed graph (digraph) whose nodes are labelled by \(1,\ldots, n\). We assume that \(G\) is unilaterally connected (a digraph is unilaterally connected if for any two arbitrary nodes \(i\) and \(j\), there exists a directed path from \(i\) to \(j\) or \(j\) to \(i\)). This implies that the zero eigenvalue of \(L_{G}\) is simple \cite{Veerman2020}. It is also an easy consequence of Gershgorin theorem that the real parts of all the non-zero eigenvalues of \(L_G\) are positive. Let \(\lambda_{1},\ldots, \lambda_{n}\) be the eigenvalues of \(L_G\), ordered according to their real parts, i.e.
\begin{equation*}
0 = \lambda_{1} < \mathrm{Re}\left(\lambda_{2}\right) \leq \mathrm{Re}\left(\lambda_{3}\right) \leq \cdots \leq \mathrm{Re}\left(\lambda_{n}\right).
\end{equation*}
The second minimum (with respect to the real-part ordering) eigenvalue, i.e. \(\lambda_{2}\), is called the spectral gap of \(G\). In this paper, we are interested in how modifying \(G\) can affect its spectral gap for models I, II and III.

Let us start with model I (see Figure \ref{Fig345uru5rtjyrdtrxgdxw}). We define
\begin{definition}\label{Defn75u6ru6ru452e5}
Consider arbitrary integers \(n\geq 3\) and \(m\geq 4\), and an arbitrary real number \(w\geq 0\). Then
\begin{enumerate}
\item for any integer \(0 \leq l \leq n\), we define 
\begin{equation}\label{eq942hk2tyaurg4frgaeyra}
\alpha_{l}:= 2\left(1-\cos\frac{l\pi}{n}\right).
\end{equation}
\item we define \(\beta_{m, w}^{-}\) and \(\beta_{m, w}^{+}\) as the roots of the quadratic polynomial \(\lambda^{2} - \left(m+w\right)\lambda + w\), i.e.
\begin{equation}\label{eq43eo8qu84t43t45t}
\beta_{m, w}^{\pm} = \frac{1}{2}\left[m + w \pm \sqrt{\left(m+w\right)^{2}-4w}\right]. 
\end{equation}
\end{enumerate}
\end{definition}

\begin{remark}
By virtue of Taylor's theorem, we can approximate \(\alpha_{l}\) for sufficiently small \(\frac{l}{n}\) by \(\alpha_{l} \approx \frac{l^{2}\pi^{2}}{n^{2}}\). Regarding \(\beta^{\pm}_{m}\), when \(\left(m+w\right)^{2} \gg 4w\), we can approximate \(\beta^{+}_{m}\) by \(m+w\), and \(\beta^{-}_{m}\) by
\begin{equation}
\beta_{m, w}^{-} = \frac{\beta_{m, w}^{-} \beta_{m, w}^{+}}{\beta_{m, w}^{+}} \approx \frac{w}{m+w}.
\end{equation}
\end{remark}
Before we proceed to our first result, let us give some intuition about this definition. The parameter $w$ in $\beta^{\pm}_{m, w}$ stands for the sum of the weights of all the cutset edges starting from the cycle and ending at the star. In the case of model I and II, we assume \(w = 1\), but for model III, we deal with arbitrary \(w\). As it is shown later (see Proposition \ref{Prop87986kiuglru3aw3er}), the spectrum of the unmodified Laplacian \(L_G\) is \( \{\alpha_{l}: \mathrm{ where }\,\, 0\leq l \leq n \,\,\mathrm{ and }\,\, l \,\,\mathrm{ is \,\,even}\} \cup \lbrace \beta_{m, w}^{-}, 1, \beta_{m, w}^{+}\rbrace\). Thus, the spectral gap of \(L_{G}\) is given by \(\min \{\alpha_{2}, \beta^{-}_{m,w}\}\). Although the \(\alpha_{l}\)s for odd \(l\) do not appear as the eigenvalues of \(L_G\), they play an important role in our theory. In particular, \(\alpha_{1}\) appears in the formulation of all the three main results of this paper.\\

Here is our main result on model I:
\begin{theoremA}[Model I]
Assume \(\beta_{m, 1}^{-}\notin\{\alpha_{l}:\, 0\leq l \leq n\}\). Consider an arbitrary modification \(\delta_{0} > 0\) and the corresponding Laplacian \(L_{G_{p}} = L_{G}\left(\delta_{0}\right)\). Then, all the eigenvalues of \(L_{G_{p}}\) are real. Moreover, we have
\begin{enumerate}[(i)]
\item if \(\alpha_{1} < \beta^{-}_{m, 1}\), then \(\lambda_{2}\left(L_{G_{p}}\right) < \lambda_{2}\left(L_{G}\right)\).
\item if \(\beta^{-}_{m, 1} < \alpha_{1}\), then \(\lambda_{2}\left(L_{G_{p}}\right) > \lambda_{2}\left(L_{G}\right)\).
\end{enumerate}
\label{Thm838o8rauger}
\end{theoremA}

\begin{remark}
The assumption \(\beta_{m, 1}^{-}\notin\{\alpha_{l}:\, 0\leq l<n\}\) in this theorem (and also in the next theorem) typically holds for arbitrary \(m\) and \(n\).
\end{remark}

Let us now discuss model II (see Figure \ref{Fig97o8gy7tuc7itca6rekrkwuayes}). Let \(\delta_{i}\geq 0\) be the weight of the edge starting from node \(i\) (see Figure \ref{Figkiut765j6ra8rekyefk7wrkweewdfsf}). Thus, model II is reduced to model I by setting \(\delta_{i} = 0\) for \(i=1,\ldots, m-1\). In this strand, we define

\begin{definition}\label{Defn76u7ku6rjtykjyts}
Let \(\delta_{i}\geq 0\), \(i= 0, \ldots, m-1\), be the weight of the modification edge starting from node \(i\) of the star and ending at node \(0\) of the cycle. We define \(\overline{\delta} := (\delta_{0}, \ldots, \delta_{m-1})\), and \(\delta := \delta_{0} + \delta_{1} + \cdots + \delta_{m-1}\).
\end{definition}

Obviously, \(\overline{\delta} = 0\) if and only if \(\delta = 0\). Note also that \(\delta = 0\) corresponds to the unmodified graph \(G\). We now state our next main result:

\begin{theoremB}[Model II]
Assume \(\beta_{m, 1}^{-}\notin\{\alpha_{l}:\, 0\leq l<n\}\). Consider a modification \(\overline{\delta} \neq 0\) and let \(L_{G_{p}} = L_{G}\left(\overline{\delta}\right)\) be the corresponding Laplacian. Then, the following hold.
\begin{enumerate}[(i)]
\item\label{Item46gktcytrxjt00} (Local modification) Let \(\overline{\delta}\neq 0\) be a sufficiently small modification. Then, all the eigenvalues of \(L_{G_{p}}\) are real, and
\begin{enumerate}
\item \label{Item46gktcytrxjt11} If \(\alpha_{1} < \beta^{-}_{m, 1}\), then \(\lambda_{2}\left(L_{G_{p}}\right) < \lambda_{2}\left(L_{G}\right)\).
\item\label{Item46gktcytrxjt33} If \(\beta^{-}_{m, 1} < \alpha_{1}\), then \(\lambda_{2}\left(L_{G_{p}}\right) > \lambda_{2}\left(L_{G}\right)\).
\end{enumerate}
\item\label{Item46gktcytrxjt44} (Global modification) Let \(\overline{\delta}\neq 0\) be an arbitrary modification. We have
\begin{enumerate}
\item\label{Item46gktcytrxjt55} If \(\alpha_{2} < \beta^{-}_{m, 1}\), then \(\mathrm{Re}\left(\lambda_{2}\left(L_{G_{p}}\right)\right) < \lambda_{2}\left(L_{G}\right)\).
\item \label{Item46gktcytrxjt22} Assume the condition \(\delta < \delta_{0}\beta_{m,1}^{+}\) is satisfied. Then, all the eigenvalues of \(L_{G_{p}}\) are real, and the statements (\ref{Item46gktcytrxjt11}) and (\ref{Item46gktcytrxjt33}) of this theorem also hold for the modification \(\overline{\delta}\).
\end{enumerate}
\end{enumerate}
\label{Theorem18na75t87ckq7tyrm3wa}
\end{theoremB}

This theorem is proved in Section \ref{Sec378uty7tdguw}. Let us mention a few remarks.

\begin{remark}
Note that, by setting \(\delta = \delta_{0}\), Theorem \ref{Thm838o8rauger}  directly follows from Theorem \ref{Theorem18na75t87ckq7tyrm3wa}.
\end{remark}

\begin{remark}
In spite of Theorem \ref{Thm838o8rauger} for which the main statements hold for a modification of arbitrary size, in Theorem \ref{Theorem18na75t87ckq7tyrm3wa}, we require a condition on the modification, i.e. \(\delta < \delta_{0}\beta_{m,1}^{+}\), to make the statements for modifications of arbitrary size. Roughly speaking, this is due to the possibility of the emergence of non-real eigenvalues. Indeed, as it is shown in the proof of Theorem \ref{Theorem18na75t87ckq7tyrm3wa}, for small modification \(\overline{\delta} \neq 0\), the modified Laplacian \(L_{G_{p}}\) has two real eigenvalues in the interval \((\alpha_{n-1}, \infty)\). However, as \(\overline{\delta}\) varies and gets larger in size, these two real eigenvalues may collide and become a pair of complex conjugates. In this case, we can think of the scenario in which the real part of these eigenvalues decreases such that for some sufficiently large modification \(\overline{\delta}\), these eigenvalues become the spectral gap of \(L_{G_{p}}\). By assuming \(\delta < \delta_{0}\beta_{m,1}^{+}\), we indeed avoid this scenario.
\end{remark}

We now discuss model III (see Figure \ref{Fig9798v7tc5extdht4ebrdbgrstgds111}). Let \(w_{i}\geq 0\), where \(i=0,\ldots, n-1\), be the weight of the edge starting from node \(i\) of the cycle. Without loss of generality, assume \(w_{0} > 0\). We also define
\begin{definition}
Let \(w_{i}\) be as mentioned above. We define \(\overline{w} = (w_{0}, \ldots, w_{n-1})\) and \(w = w_{0} + w_{1} + \cdots + w_{n-1}\).
\end{definition}

We show later that \(\lambda_{2}\left(L_{G}\right) = \min \{\alpha_{2}, \beta^{-}_{m, w}\}\). Regarding the modification in the case of model III, we consider the same family of modifications as we considered in model II: for every \(0\leq i\leq m-1\), there exists a modification edge with weight \(\delta_{i}\geq 0\) starting from node \(i\) of the star and ending at node \(0\) of the cycle (see Figure \ref{Fig9798v7tc5extdht4ebrdbgrstgds222}). Let \(\overline{\delta}\) and \(\delta\) be as in Definition \ref{Defn76u7ku6rjtykjyts}. For given \(m\), \(n\), \(\overline{w}\), \(\delta_{0}\) and \(\delta\), in the case that \(\alpha_{2} \neq \beta^{-}_{m,w}\), we also define
\begin{equation}\label{eq285tguiuugftabigsipewsf}
S = S\left(m,n,\overline{w}, \delta_{0}, \delta\right) := \delta - \frac{\delta - \delta_{0} \alpha_{2}}{\alpha_{2}^{2} - \left(m+w\right)\alpha_{2} + w}\sum_{i=0}^{n-1} w_{i}\cos\frac{2i\pi}{n}.
\end{equation}
As it is shown later, the sign of \(S\) determines if the characteristic polynomial of \(L_{G_p}\), i.e. \(\mathrm{det}\left(L_{G_p} - \lambda I\right)\), decreases or increases at the point \(\lambda = \alpha_{2}\). Our last main result is as follows.
\begin{theoremC}[Model III] 
Assume \(\beta_{m, w}^{-}\notin\{\alpha_{l}:\, 0\leq l \leq n\}\). Consider a modification \(\overline{\delta} \neq 0\) and let \(L_{G_{p}} = L_{G}\left(\overline{\delta}\right)\) be the corresponding Laplacian. Then, the following hold.
\begin{enumerate}[(i)]
\item\label{Item456rcuarl6kuwer} (Local modification) Let \(\overline{\delta}\neq 0\) be sufficiently small. Then, all the eigenvalues of \(L_{G_{p}}\) are real, and we have
\begin{enumerate}
\item \label{Item42iyo8yaor4gtwf1} If \(\alpha_{2} < \beta^{-}_{m,w}\) and \(S < 0\), then \(\lambda_{2}\left(L_{G_{p}}\right) < \lambda_{2}\left(L_{G}\right)\).
\item \label{Item42iyo8yaor4gtwf2} If \(\alpha_{2} < \beta^{-}_{m,w}\) and \(S > 0\), then \(\lambda_{2}\left(L_{G_{p}}\right) = \lambda_{2}\left(L_{G}\right)\).
\item \label{Item75988tugqyefriyer111} If \(0 < \beta^{-}_{m,w} < \alpha_{1}\), then \(\lambda_{2}\left(L_{G_{p}}\right) > \lambda_{2}\left(L_{G}\right)\).
\item \label{Item75988tugqyefriyer222} If \(\alpha_{1} < \beta^{-}_{m,w} < \alpha_{2}\) and \(\sum_{i=0}^{n-1} w_{i}\cos\left(\frac{n}{2} - i\right)\theta > 0\), where \(\theta = \pi - \cos^{-1}\left(\frac{\beta^{-}_{m,w} -2}{2}\right)\), then \(\lambda_{2}\left(L_{G_{p}}\right) > \lambda_{2}\left(L_{G}\right)\).
\item \label{Item75988tugqyefriyer333} If \(\alpha_{1} < \beta^{-}_{m,w} < \alpha_{2}\) and \(\sum_{i=0}^{n-1} w_{i}\cos\left(\frac{n}{2} - i\right)\theta < 0\), where \(\theta = \pi - \cos^{-1}\left(\frac{\beta^{-}_{m,w} -2}{2}\right)\), then \(\lambda_{2}\left(L_{G_{p}}\right) < \lambda_{2}\left(L_{G}\right)\).
\end{enumerate}
\item\label{Item89a8grye8gfaergeww} (Global modification) Let \(\overline{\delta}\neq 0\) be an arbitrary modification and assume \(\alpha_{2} < \beta^{-}_{m,w}\).
\begin{enumerate}
\item \label{Item538youkugaliyfyejr11} If \(S < 0\), then \(\mathrm{Re}\left(\lambda_{2}\left(L_{G_{p}}\right)\right) < \lambda_{2}\left(L_{G}\right)\).
\item \label{Item538youkugaliyfyejr22} If \(S > 0\), then \(\mathrm{Re}\left(\lambda_{2}\left(L_{G_{p}}\right)\right) \leq \lambda_{2}\left(L_{G}\right)\).
\end{enumerate}
\end{enumerate}
\label{Thm274i7tigfdyf}
\end{theoremC}

\section{The Laplacian matrices}\label{LaplacianMatrices}

\subsection{The Laplacian $L_{G}$ of the unmodified graph and its spectrum}\label{Sec64o8ov7t52hyd}

In this section, we investigate the spectrum of the unmodified Laplacian matrix \(L_{G}\). Denote the Laplacian matrices of the cycle \(C_{n}\) and the star \(S_{m}\) by \(L_{C_{n}}\) and \(L_{S_{m}}\), respectively. Then
\begin{equation} \label{triangularstructure}
L_{G} := \left(\begin{array}{cc}
L_{C_n} & 0\\
-C & L_{S_m} + D_{C}
\end{array}\right),
\end{equation}
where
\begin{equation*}
L_{C_{n}} = \left(\begin{array}{cccccc}
2 & -1 & & -1\\
-1 & 2 & \ddots & \\
 & \ddots & \ddots & -1\\
-1 & & -1 & 2
\end{array}\right) \qquad \mathrm{and}\qquad L_{S_{m}} = \left(\begin{array}{cc}
m-1 & -\mathbf{1}_{m-1}^{\top}\\
-\mathbf{1}_{m-1} & I_{m-1}
\end{array}\right).
\end{equation*}
Moreover, for models I and II, we have 
\begin{equation}\label{eq998yo97h1dkigeodvekfge22kfhqge}
C=\left(\begin{array}{cc}
1 & 0_{1\times\left(n-1\right)}\\
0_{\left(m-1\right)\times 1} & 0_{\left(m-1\right)\times \left(n-1\right)}
\end{array}\right) \quad \mathrm{and}\quad
D_{C} =\left(\begin{array}{cc}
1 & 0_{1\times\left(m-1\right)}\\
0_{\left(m-1\right)\times 1} & 0_{\left(m-1\right)\times \left(m-1\right)}
\end{array}\right),
\end{equation}
and for model III, we have
\begin{equation}\label{eq801yjhyjfdufe87t32pg1eo30uhus}
C = \left(\begin{array}{c}
\begin{array}{cccc}
w_{0} & w_{1} & \cdots & w_{n-1}
\end{array}\\
0_{(m-1)\times n}
\end{array}\right)\quad\mathrm{and}\quad D_{C} =\left(\begin{array}{cc}
w & 0_{1\times\left(m-1\right)}\\
0_{\left(m-1\right)\times 1} & 0_{\left(m-1\right)\times \left(m-1\right)}
\end{array}\right).
\end{equation}

The block triangular form of \(L_{G}\) implies \(\sigma(L_{G}) = \sigma(L_{C_{n}}) \cup \sigma(L_{S_{m}}+D_{C})\). Thus, to study \(\sigma(L_{G})\), we need to investigate each of \(\sigma(L_{C_{n}})\) and \(\sigma(L_{S_{m}}+D_{C})\) individually. In this strand, we have the following lemmas.

\begin{lemma}\label{Lem616uihugygjyhvkjgoigsaougsiky}
Recall Definition \ref{Defn75u6ru6ru452e5}. We have \(\sigma(L_{C_{n}}) = \{\alpha_{l}: \mathrm{where}\,\, 0\leq l \leq n \,\,\mathrm{and}\,\, l\,\, \mathrm{is\,\, even}\}\). Moreover, the multiplicity of all the eigenvalues except for \(0\) and \(4\) (the eigenvalue \(4\) appears only when \(n\) is even) is \(2\).
\end{lemma}
\begin{proof}
See \cite{Brouwer2011spectra}.
\end{proof}

\begin{lemma}
Let \(C\) and \(D_{C}\) be as in (\ref{eq801yjhyjfdufe87t32pg1eo30uhus}). Then, \(\sigma(L_{S_{m}}+D_{C}) = \lbrace \beta_{m, w}^{-}, 1, \beta_{m, w}^{+}\rbrace\), where \(\beta_{m, w}^{\pm}\) are as in (\ref{eq43eo8qu84t43t45t}). Moreover, the eigenvalues \(\beta_{m, w}^{-}\) and \(\beta_{m, w}^{+}\) are simple, and the eigenvalue \(1\) is of multiplicity \(m-2\).
\end{lemma}
\begin{proof}
This lemma is a special case of Lemma \ref{Lem391q0kieb6fb2j2iwbjyfu5ewwe}, which is proved in \ref{Section4311ohbiwubwavs}.
\end{proof}

The previous two lemmas give the spectrum of the unmodified Laplacian \(L_{G}\):
\begin{proposition}\label{Prop87986kiuglru3aw3er}
We have \(\sigma(L_{G}) = \{\alpha_{l}: \mathrm{ where }\,\, 0\leq l \leq n \,\,\mathrm{ and }\,\, l \,\,\mathrm{ is \,\,even}\} \cup \lbrace \beta_{m, w}^{-}, 1, \beta_{m, w}^{+}\rbrace\).
\end{proposition}

\begin{remark}\label{Rem95jyru6rmtdhtlirstw}
We assume that \(m\geq 4\), i.e. the star \(S_m\) has at least four nodes. It is straightforward to show that for any \(m\geq 4\) and \(w > 0\), we have \(\beta_{m, w}^{-} < 1\) and \(4 < \beta_{m, w}^{+}\). On the other hand, \(0\leq \alpha_{l} = 2\left(1-\cos \frac{l\pi}{n}\right) \leq 4\), for all \(0\leq l \leq n\). This means that \(\beta^{+}_{m,w}\) is a simple eigenvalue of \(L_{G}\).
\end{remark}

\subsection{The Laplacian $L_{G_p}$ of the modified graph}

Consider model III and observe that the modified Laplacian matrix \(L_{G_{p}}\) is given by
\begin{equation}\label{eq6g7k6iq75ckt7rer}
L_{G_p} := \left(\begin{array}{cc}
L_{C_n} + D_{\Delta} & -\Delta\\
-C & L_{S_m} + D_{C}
\end{array}\right),
\end{equation}
where \(C\) and \(D_{C}\) are as in (\ref{eq801yjhyjfdufe87t32pg1eo30uhus}),
\begin{equation}\label{eq72r887tiygrjfcyr4}
\Delta = \left(\begin{array}{c}
\begin{array}{cccc}
\delta_{0} & \delta_{1} & \cdots & \delta_{m-1}
\end{array}\\
0_{\left(n-1\right)\times m}
\end{array}\right)\qquad
\mathrm{ and }\qquad
D_{\Delta} =\left(\begin{array}{cc}
\delta & 0_{1\times\left(n-1\right)}\\
0_{\left(n-1\right)\times 1} & 0_{\left(n-1\right)\times \left(n-1\right)}
\end{array}\right).
\end{equation}

\begin{notation}\label{Notation94kyvtw6r4eu6t3r}
For the sake of convenience, we set \(L_{1}:= L_{C_{n}} + D_{\Delta}\) and \(L_{2} := L_{S_m} + D_{C}\).
\end{notation}
Using this notation, Laplacian (\ref{eq6g7k6iq75ckt7rer}) is written as
\begin{equation}\label{eq9786itgkug54ourgkuw4geugf}
L_{G_{p}} = \left(\begin{array}{cc}
L_{1} & -\Delta\\ -C & L_2
\end{array}\right).
\end{equation}

The Laplacian \(L_{G_{p}}\) of the modified graph of model II is of the form (\ref{eq9786itgkug54ourgkuw4geugf}), where \(C\) and \(D_{C}\) are as in (\ref{eq998yo97h1dkigeodvekfge22kfhqge}), and \(\Delta\) and \(D_{\Delta}\) are given by (\ref{eq72r887tiygrjfcyr4}).

The Laplacian \(L_{G_{p}}\) of the unmodified graph of model I is also of the form (\ref{eq9786itgkug54ourgkuw4geugf}), where \(C\) and \(D_{C}\) are as in (\ref{eq998yo97h1dkigeodvekfge22kfhqge}), and \(\Delta\) and \(D_{\Delta}\) are given by
\begin{equation}\label{eq9kf5utu64u6rut3e}
\Delta = \left(\begin{array}{cc}
\delta_{0} & 0_{1\times\left(m-1\right)}\\
0_{\left(n-1\right)\times 1} & 0_{\left(n-1\right)\times \left(m-1\right)}
\end{array}\right)\qquad
\mathrm{ and }\qquad
D_{\Delta} = \left(\begin{array}{cc}
\delta & 0_{1\times\left(n-1\right)}\\
0_{\left(n-1\right)\times 1} & 0_{\left(n-1\right)\times \left(n-1\right)}
\end{array}\right).
\end{equation}
Here (model I), we have \(\delta_{0} = \delta\).

Notice that, in all these three models, despite the unmodified Laplacian \(L_{G}\), the modified Laplacian \(L_{G_{p}}\) does not have a triangular form. Due to this reason, analysis of the spectrum of \(L_{G_p}\) requires further work. We deal with this analysis in the next section.

\section{Proofs of the main results}

In this section, we prove our main results: Theorems \ref{Theorem18na75t87ckq7tyrm3wa} and \ref{Thm274i7tigfdyf} (Theorem \ref{Thm838o8rauger} follow from Theorem \ref{Theorem18na75t87ckq7tyrm3wa}). Note that model II can be considered as a special case of model III. Thus, it is reasonable to introduce the main concepts and notations of the proofs in this section mainly based on model III. This section is organized as follows. We first discuss some preliminaries, definitions and notations in Section \ref{Sec902u3h2be}. In Section \ref{Sec6o87tg2eigwiugd}, we discuss the techniques that are used in the proofs of the theorems. We then prove Theorem \ref{Theorem18na75t87ckq7tyrm3wa} in Section \ref{Sec378uty7tdguw}. Finally, we prove Theorem \ref{Thm274i7tigfdyf} in Section \ref{Sec12384o8y8agrit4}.

\subsection{Preliminaries, definitions and notations}\label{Sec902u3h2be}

In this section, we discuss some preliminaries, and introduce some concepts and notations which are used throughout the proofs.
\begin{notation}
Throughout, \(\mathbf{1}_{k}\) stands for the \(k\)-dimensional vector whose entries are all \(1\). We may drop \(k\) when it is clear from the context.
\end{notation}

\begin{definition}\label{Defn537oiyb8yavoyvyfceluiftluitfa}
Let \(w\), \(\delta_{0}\) and \(\delta\) be real, and \(m\) and \(k\) be positive integers. Consider \(\lambda\in\mathbb{R}\).
\begin{enumerate}[(i)]
\item We define \(\mu: \lambda \mapsto \mu (\lambda)\) by
\begin{equation}\label{eq88yzz8qr7t4rygfea}
\mu = \mu (\lambda) = \frac{1- \lambda}{\lambda^{2} - \left(m+w\right)\lambda + w},
\end{equation}
and \(y: \lambda \mapsto y(\lambda)\) by
\begin{equation}\label{e2e1q7983q75c6r3ytfurwrks}
y = y(\lambda) = \frac{\delta - \delta_{0} \lambda}{\lambda^{2} - \left(m+w\right)\lambda + w}.
\end{equation}
\item For any \(k\geq 3\), we define
\begin{equation}\label{eq13587itu1v6rtkuyrj6ytdcs}
Q_{k} = Q_{k}\left(\lambda\right) = \left(\begin{array}{ccccc}
\lambda - 2 & 1\\
1 & \lambda - 2 & 1 & & 0\\
 & 1 & \ddots & \ddots \\
0 & & \ddots & \ddots & 1\\
 & & & 1 & \lambda-2
\end{array}\right)_{k\times k}.
\end{equation}
\end{enumerate}
\end{definition}

The next two lemmas investigate the matrix \(Q_{k}(\lambda)\) for different values of \(\lambda > 0\). See \cite{hu1996analytical} for the proofs\footnote{Regarding Lemma \ref{Lem092uy7tcrmyftdrrdfdfdfw}, the formulas in \cite{hu1996analytical} are not totally correct. In this current paper, we have used the corrected ones. Note also that \(\theta\) in this current paper is not the same as in \cite{hu1996analytical}}.

\begin{lemma}\label{Lem092uy7tcrmyftdrrdfdfdfw}
Assume \(0 < \lambda < 4\) and let \(\theta = \pi - \cos^{-1}(\frac{\lambda -2}{2})\). We have
\begin{enumerate}[(i)]
\item \(\mathrm{det}(Q_{k}) = \frac{(-1)^{k}\sin\left(k+1\right)\theta}{\sin \theta}\).
\item the matrix \(R = Q_{k}^{-1}\) exists for \(\theta \neq \frac{l\pi}{k+1}\) (\(l = 1, \ldots, k\)), and is given by
\begin{equation}\label{eq49nhuotqi3f2q3dut3drt11}
R_{ij} = \frac{\cos\left(k+1-\left\vert i-j\right\vert\right)\theta - \cos\left(k+1-i-j\right)\theta}{2\sin \theta \sin\left(k+1\right)\theta}, \qquad \mathrm{for\,\, } 1\leq i, j\leq k.
\end{equation}
\end{enumerate}
\end{lemma}

\begin{lemma}\label{Lem092uy7tcriugygeyavdjghesvfw}
Assume \(\lambda \geq 4\) and let \(\theta = \cosh^{-1}(\frac{\lambda -2}{2})\). Then
\begin{enumerate}[(i)]
\item for \(\lambda > 4\), we have \(\mathrm{det}(Q_{k}) = \frac{\sinh\left(k+1\right)\theta}{\sinh \theta}\).
\item for \(\lambda = 4\), we have \(\mathrm{det}(Q_{k}) = k+1\).
\item The inverse matrix \(R = Q_{k}^{-1}\) exists for all \(\lambda \geq 4\), and is given by
\begin{equation}\label{eq49nhuotqi3f2q3dut3drt22}
R_{ij} = \left(-1\right)^{i+j} \cdot \frac{\cosh\left(k+1-\left\vert i-j\right\vert\right)\theta - \cosh\left(k+1-i-j\right)\theta}{2\sinh \theta \sinh\left(k+1\right)\theta}, \qquad \mathrm{for\,\, } 1\leq i, j\leq k.
\end{equation}
\end{enumerate}
\end{lemma}

Recall \(\alpha_{l}\) defined by (\ref{eq942hk2tyaurg4frgaeyra}). By Lemmas \ref{Lem092uy7tcrmyftdrrdfdfdfw} and \ref{Lem092uy7tcriugygeyavdjghesvfw}, and a straightforward calculation, we have
\begin{lemma}\label{Lem42kutkbqifop23ifbwe}
The matrix \(Q_{n-1}(\lambda)\) is invertible if and only if \(\lambda \neq \alpha_{l}\) for \(l = 1,\ldots, n-1\).
\end{lemma}

\subsection{Our approach for investigating the spectrum of the modified Laplacian $L_{G_{p}}$}\label{Sec6o87tg2eigwiugd}

In this section, we discuss the method we use to investigate the spectrum of the modified Laplacian \(L_{G_p}\). We directly apply this method to study model III and then use the results to investigate models I and II.

Recall that the modified Laplacian of model III is given by
\begin{equation}\label{eq8i6vitkur5r4efnhr5e}
L_{G_{p}} = \left(\begin{array}{cc}
L_{1} & -\Delta\\ -C & L_2
\end{array}\right),
\end{equation}
where \(L_{1}\) and \(L_{2}\) are as in Notation \ref{Notation94kyvtw6r4eu6t3r}, and the matrices \(C\) and \(\Delta\) are given by (\ref{eq801yjhyjfdufe87t32pg1eo30uhus}) and
(\ref{eq72r887tiygrjfcyr4}), respectively. Our study of the eigenvalues of \(L_{G_{p}}\) is based on the following lemma.

\begin{lemma}\label{Lem4401oihe8uoihh29290ihbcd}
Consider the modified Laplacian \(L_{G_p}\) given by (\ref{eq8i6vitkur5r4efnhr5e}). For \(\lambda \in \mathbb{R}\), we have
\begin{enumerate}[(i)]
\item if \(\lambda \notin \sigma(L_{1})\), then \(\mathrm{det}\left(L_{G_p} -\lambda I\right) = \mathrm{det}\left(L_{1} -\lambda I\right) \cdot P_{1}\left(\lambda\right)\), where \(P_{1}(\lambda) = \mathrm{det}(M_{1})\), for \(M_{1} = M_{1}(\lambda) = L_{2}-\lambda I - C \left(L_{1} -\lambda I\right)^{-1} \Delta\).
\item if \(\lambda\notin \sigma(L_{2})\), then \(\mathrm{det}\left(L_{G_p} -\lambda I\right) = \mathrm{det}\left(L_{2} - \lambda I\right) \cdot P_{2}\left(\lambda\right)\), where \(P_{2}(\lambda) = \mathrm{\det}(M_2)\), for \(M_2 = M_{2}(\lambda) = L_{1} -\lambda I - \Delta \left(L_{2} -\lambda I\right)^{-1} C\).
\item for \(i=1,2\), we have that \(\lambda_{0}\notin \sigma(L_i)\) is an eigenvalue of \(L_{G_p}\) with algebraic multiplicity \(k\), if and only if \(P_{i}(\lambda_{0}) = P^{\prime}_{i}(\lambda_{0}) = \cdots = \frac{d^{k-1} P_{i}}{d \lambda^{k-1}}(\lambda_0) = 0\), and  \(\frac{d^{k} P_{i}}{d \lambda^{k}}(\lambda_0)\neq 0\).
\end{enumerate}
\end{lemma}

\begin{remark}\label{Rem5092hgyfjngctdhgcgcs}
Lemma \ref{Lem4401oihe8uoihh29290ihbcd} allows us to count the multiplicity of \(\lambda_{0}\in \sigma(L_{G_p})\) when \(\lambda_{0}\notin \sigma(L_1) \cap \sigma(L_{2})\). However, this lemma may give information about the multiplicity of \(\lambda_{0}\) when \(\lambda_{0}\in \sigma(L_1) \cap \sigma(L_{2})\) as well. This is important for us since we have such eigenvalues in our models. Let \(\lambda_{0}\) be such an eigenvalue. Since \(\lambda_{0}\in \sigma(L_1)\), the matrix \((L_1 - \lambda_{0}I)^{-1}\) does not exist. However, depending on the matrices \(C\) and \(\Delta\), the expression \(\lim_{\lambda \rightarrow \lambda_{0}} Y(\lambda)\), where \(Y(\lambda):= C(L_1 - \lambda_{0}I)^{-1}\Delta\), may exist. This allows us to define \(M_{1}\) and \(P_{1}\) at \(\lambda = \lambda_{0}\) by taking the limit \(\lambda\rightarrow \lambda_{0}\). Now, if \(Y(\lambda)\) at \(\lambda = \lambda_{0}\) is smooth enough, then the multiplicity of \(\lambda_{0}\) as an eigenvalue of \(L_{G_p}\) is \(l + k\), where \(l\) is the multiplicity of \(\lambda_{0}\) as an eigenvalue of \(L_{1}\) and \(k\) is the integer that satisfies \(P_{1}(\lambda_{0}) = P^{\prime}_{1}(\lambda_{0}) = \cdots = \frac{d^{k-1} P_{1}}{d \lambda^{k-1}}(\lambda_0) = 0\), and \(\frac{d^{k} P_{1}}{d \lambda^{k}}(\lambda_0)\neq 0\). Analogous holds when \(\lambda_{0}\in \sigma(L_2)\) but \(\Delta \left(L_{2} -\lambda I\right)^{-1} C\) is well-defined and smooth enough at \(\lambda = \lambda_{0}\).
\end{remark}

According to Lemma \ref{Lem4401oihe8uoihh29290ihbcd}, an eigenvalue \(\lambda\) of \(L_{G_p}\) that is not in \(\sigma(L_1) \cap \sigma(L_{2})\) must satisfy \(P_{1}(\lambda) = 0\) or \(P_{2}(\lambda) = 0\). The proofs of our results are based on the analysis of these two equations. Sections \ref{Section2i7tiug7tyge} and \ref{Section1782kli0jyf37uyf} are dedicated to this analysis.

Before we proceed further, let us show that \(\lambda = 1\) is an eigenvalue of \(L_{G_{p}}\) for any arbitrary \(\overline{\delta}\).
\begin{lemma}\label{Lem36iyguyfhctgczhvjgwi}
For arbitrary \(\overline{\delta}\), we have \(1\in\sigma(L_{G_{p}})\). Moreover, the (algebraic and geometric) multiplicity of \(1\) is at least \(m-2\).
\end{lemma}
\begin{proof}
Recall that \(L_{G_{p}} = \left(\begin{smallmatrix} L_{1} & -\Delta\\
-C & L_2\end{smallmatrix}\right)\). It follows from the proof of Lemma \ref{Lem391q0kieb6fb2j2iwbjyfu5ewwe} (see relation (\ref{eq80279867tuyfwytchgcss})) that there exist \(m-2\) linearly independent left eigenvectors \(v\) such that \(v^{\top} L_{2} = v^{\top}\). Moreover, any such a vector \(v\) is of the form \(v = (0, v_{1}, \cdots, v_{m-1})\in \mathbb{R}^{m}\) (the first entry is zero). Consider the vector \(u:= (0, v)\in\mathbb{R}^{n+m}\). Taking into account that, except for the first row, all the entries of \(C\) are zero (see (\ref{eq801yjhyjfdufe87t32pg1eo30uhus})), we obtain
\begin{equation*}
u^{\top} L_{G_{p}} = \left(0_{1\times n}, v^{\top}\right) \left(\begin{array}{cc} L_{1} & -\Delta\\ -C & L_2\end{array}\right) = \left(0_{1\times n}, v^{\top} L_{2}\right) = u^{\top}.
\end{equation*}
This means that for such \(v\)s, the corresponding vectors \(u\) are left eigenvectors of \(L_{G_{p}}\) associated with the eigenvalue \(1\). This proves the lemma.
\end{proof}

\subsubsection{Analysis of $P_{2}$}\label{Section2i7tiug7tyge}

In this section, we investigate the matrix \(M_{2}(\lambda)\) and the function \(P_{2}(\lambda):= \mathrm{det}(M_{2}(\lambda))\) introduced in Lemma \ref{Lem4401oihe8uoihh29290ihbcd} for model III. We first need to analyze the matrix \(L_{2} - \lambda I\) and its inverse:
\begin{lemma}\label{Lem52o08y8y54667tytuyfd}
Recall \(\mu\) from (\ref{eq88yzz8qr7t4rygfea}). We have
\begin{enumerate}[(i)]
\item \label{Item9893r8i} the function \(\mu\) is well-defined at \(\lambda\notin\lbrace \beta^{-}_{m, w}, \beta^{+}_{m, w}\rbrace\).
\item \label{Item9893r8ii} for \(\lambda\in\mathbb{R}\setminus\sigma(L_{2}) = \lbrace \beta^{-}_{m, w}, 1, \beta^{+}_{m, w}\rbrace\), we have
\begin{equation}\label{eq118y8qbt3447rg4tgt}
\left(L_{2} - \lambda I\right)^{-1} = 
\left(\begin{array}{cc}
m -1 + w - \lambda & -\mathbf{1}^{\top}\\
-\mathbf{1} & \left(1-\lambda\right) I
\end{array}\right)^{-1} = \left(\begin{array}{cc}
\mu & \frac{\mu}{1-\lambda}\mathbf{1}^{\top}\\
\frac{\mu}{1-\lambda}\mathbf{1} & \frac{1}{1-\lambda} I + \frac{\mu}{\left(1-\lambda\right)^{2}}\mathbf{1}\mathbf{1}^{\top}
\end{array}\right).
\end{equation}
\end{enumerate}
\end{lemma}
\begin{proof}
Item (\ref{Item9893r8i}) is straightforward. Item (\ref{Item9893r8ii}) follows from Lemma \ref{lem11y78ty998796f56drytr}.
\end{proof}

We now start to calculate \(M_2 = M_{2}(\lambda) = L_{1} -\lambda I - \Delta \left(L_{2} -\lambda I\right)^{-1} C\). The expression \(\left(L_{2} - \lambda I\right)^{-1}\) is well-defined at \(\lambda\notin\sigma(L_{2}) = \lbrace \beta^{-}_{m, w}, 1, \beta^{+}_{m, w}\rbrace\). By a straightforward calculation and using relation (\ref{eq118y8qbt3447rg4tgt}), for \(\lambda\notin\sigma(L_{2})\), we have \(\Delta \left(L_{2} -\lambda I\right)^{-1} C = y C\), where \(y = y(\lambda)\) is given by (\ref{e2e1q7983q75c6r3ytfurwrks}). Note that \(y\), and therefore \(yC\), is well-defined and smooth at \(\lambda = 1\). In other words, although \(\left(L_{2} -\lambda I\right)^{-1}\) is not defined at \(\lambda = 1\) (because \(1\in\sigma(L_{2})\)), the expression \(\Delta \left(L_{2} -\lambda I\right)^{-1} C\) can be defined at \(\lambda = 1\), and so do the matrix \(M_{2}\) and the function \(P_{2}\). This was discussed earlier in Remark \ref{Rem5092hgyfjngctdhgcgcs}. We give the following lemma to emphasize this property.

\begin{lemma}\label{Lemy9yq7g3ryg4frewiyjfre}
The function \(P_{2}(\lambda) = \mathrm{det}(M_{2})\) is well-defined and smooth at \(\lambda\notin\lbrace \beta^{-}_{m, w}, \beta^{+}_{m, w}\rbrace\).
\end{lemma}

Having \(\Delta \left(L_{2} -\lambda I\right)^{-1} C = y C\), we obtain
\begin{equation}\label{eq9u8u0huqeiuwgewa218}
M_2 = M_2\left(\lambda\right) = \left(\begin{array}{c|c}
2-\lambda + \delta - w_{0}y &
\begin{array}{ccccc}
-1-w_{1}y & -w_{2}y & \cdots & w_{n-2}y & -1-w_{n-1}y
\end{array}
\\
\hline
\begin{array}{c}
-1\\ 0_{(n-3)\times 1} \\ -1
\end{array} & -Q_{n-1}
\end{array}\right),
\end{equation}
where \(Q_{n-1} = Q_{n-1}(\lambda)\) is the symmetric tridiagonal matrix given by (\ref{eq13587itu1v6rtkuyrj6ytdcs}). Applying Lemma \ref{lem11y78ty998796f56drytr} on this matrix, for \(\lambda\notin\lbrace\beta^{-}_{m, w}, \beta^{+}_{m, w}\rbrace\) such that \(Q_{n-1}(\lambda)\) is invertible (recall that, by Lemma \ref{Lem42kutkbqifop23ifbwe}, the matrix \(Q_{n-1}(\lambda)\) is invertible if and only if \(\lambda \neq \alpha_{l}\) for \(l = 1,\ldots, n-1\)), we obtain
\begin{equation}\label{eq6545eytd7wicq67rewq}
P_{2}\left(\lambda\right) = \mathrm{det}\left(M_2\right) = \left(-1\right)^{n-1} \mathrm{det}\left(Q_{n-1}\right) \left[\xi\left(\delta, \lambda\right) - y \psi\left(w,\lambda\right)\right],
\end{equation}
where
\begin{equation}\label{eqyb2ydvpoetvefigye2wyf2ew}
\xi\left(\delta, \lambda\right) = 2-\lambda + \delta + R_{11} + R_{1\, n-1} + R_{n-1\, 1} + R_{n-1\, n-1},
\end{equation}
for which \(R = \left(R_{ij}\right)_{1\leq i, j\leq n-1}\) is the inverse of \(Q_{n-1}\), and
\begin{equation}\label{eq04gaigf83f42dhgfa7973}
\psi = \psi\left(\overline{w}, \lambda\right) = w_{0} - \sum_{i=1}^{n-1} w_{i}\left[R_{i1}+ R_{i\, n-1}\right].
\end{equation}
Lemmas \ref{Lem092uy7tcrmyftdrrdfdfdfw} and \ref{Lem092uy7tcriugygeyavdjghesvfw} give some formulas for \(R=Q_{n-1}^{-1}\). Substituting these formulas in (\ref{eqyb2ydvpoetvefigye2wyf2ew}) and (\ref{eq04gaigf83f42dhgfa7973}) gives
\begin{lemma}\label{Lem8tvi7tvli7ut6rdkawtudfkd}
For the functions \(\xi\left(\delta, \lambda\right)\) and \(\psi(\overline{w}, \lambda)\), we have
\begin{equation*}
\xi = \xi\left(\delta, \lambda\right) = \left\lbrace\begin{array}{ll}
\delta -2\sin\theta \tan\frac{n\theta}{2}, & 0 < \lambda < 4 \mathrm{\,\, and\,\, } \theta = \pi - \cos^{-1}(\frac{\lambda -2}{2}),\\
\delta + \frac{2}{n}\cdot\left[\left(-1\right)^{n} - 1\right], & \lambda = 4, \\
\delta + \frac{2\sinh\theta}{\sinh n\theta}\cdot\left[\left(-1\right)^{n} - \cosh n\theta\right],\quad & \lambda >4 \mathrm{\,\, and \,\,} \theta = \cosh^{-1}(\frac{\lambda -2}{2}),
\end{array}\right.
\end{equation*}
and
\begin{equation*}
\psi = \psi\left(\overline{w}, \lambda\right) = \left\lbrace\begin{array}{ll}
\frac{1}{\cos\frac{n\theta}{2}}\sum_{i=0}^{n-1} w_{i}\cos\left(\frac{n}{2} - i\right)\theta, & 0 < \lambda < 4 \mathrm{\,\, and\,\, } \theta = \pi - \cos^{-1}(\frac{\lambda -2}{2}),\\
\sum_{i=0}^{n-1} \left(-1\right)^{i} w_{i}, & n \mathrm{\,\, is\,\, even, \,\,}\lambda = 4, \\
\sum_{i=0}^{n-1} \left(-1\right)^{i} w_{i}\left[1-\frac{2i}{n}\right], & n \mathrm{\,\, is\,\, odd, \,\,}\lambda = 4, \\
\frac{1}{\cosh\frac{n\theta}{2}}\sum_{i=0}^{n-1} \left(-1\right)^{i} w_{i}\cosh\left(\frac{n}{2} - i\right)\theta,\quad & n \mathrm{\,\, is\,\, even, \,\,}\lambda >4 \mathrm{\,\, and \,\,} \theta = \cosh^{-1}(\frac{\lambda -2}{2}), \\
\frac{1}{\sinh\frac{n\theta}{2}}\sum_{i=0}^{n-1} \left(-1\right)^{i} w_{i}\sinh\left(\frac{n}{2} - i\right)\theta,\quad & n \mathrm{\,\, is\,\, odd, \,\,}\lambda >4 \mathrm{\,\, and \,\,} \theta = \cosh^{-1}(\frac{\lambda -2}{2}).
\end{array}\right.
\end{equation*}
\end{lemma}

\begin{proof}
For \(\lambda \neq 4\), the proof is a straightforward calculation by substituting (\ref{eq49nhuotqi3f2q3dut3drt11}) and (\ref{eq49nhuotqi3f2q3dut3drt22}) into (\ref{eqyb2ydvpoetvefigye2wyf2ew}) and (\ref{eq04gaigf83f42dhgfa7973}). For the case of \(\lambda = 4\), the proof follows from taking the limit of the formulas for the cases \(\lambda \neq 4\) as \(\lambda\rightarrow 4\) and using L'H\^opital's rule.
\end{proof}

This lemma together with relation (\ref{eq6545eytd7wicq67rewq}), gives some formulas for \(P_{2}(\lambda)\) when \(P_{2}\) is well-defined (\(\lambda\notin\lbrace \beta^{-}_{m, w}, \beta^{+}_{m, w}\rbrace\)) and \(Q_{n-1}(\lambda)\) is invertible, i.e. \(\lambda \neq \alpha_{l}\) for \(l = 1,\ldots, n-1\). However, we can use (\ref{eq6545eytd7wicq67rewq}) to calculate \(P_2\) at \(\lambda = \alpha_{l}\) by taking \(\lim P_{2}(\lambda)\) as \(\lambda \rightarrow \alpha_{l}\). By this trick, we have that (\ref{eq6545eytd7wicq67rewq}) is well-defined and smooth at every real \(\lambda\notin\lbrace \beta^{-}_{m, w}, \beta^{+}_{m, w}\rbrace\).

To make the analysis of \(P_{2}\) simpler, we consider two different cases of \(0 \leq \lambda < 4\) and \(\lambda\geq 4\). For the first case, let \(\theta = \pi - \cos^{-1}(\frac{\lambda -2}{2})\), and define \(p(\theta) := P_{2}(\lambda(\theta)) = P_{2}(2[1-\cos \theta])\). For \(0 < \theta < \pi\) such that \(2[1-\cos \theta]\notin\lbrace \beta^{-}_{m, w}, \beta^{+}_{m, w}\rbrace\), this gives
\begin{equation}\label{eq7guyf7t2fuyf6f2ut23sh}
p\left(\theta\right) = 2\left[\cos n\theta - 1\right] + \delta\cdot\frac{\sin n\theta}{\sin\theta} - 2y\cdot\frac{\sin \frac{n\theta}{2}}{\sin\theta}\sum_{i=0}^{n-1} w_{i}\cos\left(\frac{n}{2} - i\right)\theta.
\end{equation}
Observe that \(p\left(0\right) = \lim_{\theta\rightarrow 0^{+}} p(\theta) = 0\). With a straightforward calculation, we can also obtain:
\begin{lemma}\label{Lem526ri76er6cp95c2l7xfwsw1}
Recall \(\alpha_{l}\) given by (\ref{eq942hk2tyaurg4frgaeyra}) and assume \(\alpha_{l} = 2\left(1-\cos\frac{l\pi}{n}\right) \notin \lbrace\beta^{-}_{m, w}, \beta^{+}_{m, w}\rbrace\), where \(l\in\mathbb{Z}\) is as specified below. Then
\begin{enumerate}[(i)]
\item for even \(0 \leq l \leq n-1\), we have \(p(\frac{l\pi}{n}) = 0\).
\item for odd \(1 \leq l \leq n-1\), we have \(p(\frac{l\pi}{n}) = -4 - \frac{2y}{\sin\frac{l\pi}{n}}\sum_{i=0}^{n-1} w_{i} \sin\frac{il\pi}{n}\).
\item for even \(1 \leq l \leq n-1\), we have
\begin{equation}\label{eq800ug4tht9q438tqherug}
p^{\prime}\left(\frac{l\pi}{n}\right) = \frac{n}{\sin \frac{l\pi}{n}}\left[\delta - y\sum_{i=0}^{n-1} w_{i}\cos\frac{il\pi}{n}\right].
\end{equation}
\end{enumerate}
\end{lemma}

\subsubsection{Analysis of $P_1$}\label{Section1782kli0jyf37uyf}

In this section, we investigate the matrix \(M_{1}\) and the function \(P_{1}(\lambda) = \mathrm{det}(M_{1}(\lambda))\) introduced in Lemma \ref{Lem4401oihe8uoihh29290ihbcd} for model III. We start with analyzing the matrix \(L_{1} - \lambda I\) and its inverse. Note that
\begin{equation}\label{eq4230yl8blatv7tlujygvdaydvz}
L_{1} - \lambda I = \left(\begin{array}{c|c}
2+\delta-\lambda & \begin{array}{ccc}
-1 & 0_{1\times \left(n-3\right)} & -1
\end{array}\\
\hline
\begin{array}{c}
-1\\
0_{\left(n-3\right)\times 1}\\
-1
\end{array} & -Q_{n-1}
\end{array}\right).
\end{equation}

\begin{lemma}\label{Lem87taor64wyfro4ypoiy}
Let \(\lambda > 0\) be real. Then, the matrix \(L_{1} -\lambda I\) is invertible if and only if \(\lambda \neq \alpha_{l}\) and \(\xi(\delta, \lambda) \neq 0\), where \(1 < l \leq n-1\) is even and \(\xi(\delta, \lambda)\) is given by (\ref{eqyb2ydvpoetvefigye2wyf2ew}).
\end{lemma}

\begin{proof}
Assume \(Q_{n-1}\) is invertible. Applying Lemma \ref{lem11y78ty998796f56drytr} on matrix (\ref{eq4230yl8blatv7tlujygvdaydvz}) gives
\begin{equation*}
\mathrm{det}\left(L_{1} -\lambda I\right) = \left(-1\right)^{n-1}\mathrm{det}\left(Q_{n-1}\right) \xi\left(\delta, \lambda\right).
\end{equation*}
This proves the lemma for the case that \(Q_{n-1}\) is invertible.

Now, we consider the case that \(Q_{n-1}\) is singular. It follows from Lemma \ref{Lem42kutkbqifop23ifbwe} that \(Q_{n-1}(\lambda)\) is invertible if and only if \(\lambda \neq \alpha_{l}\) for \(l = 1,\ldots, n-1\). Equivalently \(Q_{n-1}\) is singular if and only if \(0 < \lambda = 2\left[1-\cos\theta_{0}\right] < 4\) and \(\sin n\theta_{0} = 0\) (see also Lemmas \ref{Lem092uy7tcrmyftdrrdfdfdfw} and \ref{Lem092uy7tcriugygeyavdjghesvfw}). By virtue of Lemma \ref{Lem8tvi7tvli7ut6rdkawtudfkd}, for \(\lambda = 2\left[1-\cos\theta_{0}\right]\), we obtain
\begin{equation*}
\mathrm{det}\left(L_{1} -\lambda I\right) =\lim_{\theta\rightarrow \theta_{0}} \mathrm{det}\left(L_{1} -\lambda I\right) = \lim_{\theta\rightarrow \theta_{0}} \frac{\sin n\theta}{\sin \theta}\cdot \left(\delta -2\sin\theta \tan\frac{n\theta}{2}\right) = 2\left(\cos n\theta_{0} - 1\right).
\end{equation*}
Thus, when \(Q_{n-1}\) is singular, \(L_{1} -\lambda I\) is invertible if and only if \(\cos n\theta_{0} \neq 1\), i.e. \(\lambda \neq \alpha_{l}\) where \(2 \leq l \leq n-1\) is even. This ends the proof.
\end{proof}

For real \(\lambda > 0\), assume \(\xi\left(\delta, \lambda\right) \neq 0\) and consider the case that \(R = Q_{n-1}^{-1}\) exists. Then, \(L_{1} -\lambda I\) is invertible, and by Lemma \ref{lem11y78ty998796f56drytr}, we have
\begin{equation*}
\left(L_{1} - \lambda I\right)^{-1} = \left(\begin{array}{cc}
\xi^{-1} & -\xi^{-1} r^{\top}\\
-\xi^{-1} r & -R + \xi^{-1} r r^{\top}
\end{array}\right),
\end{equation*}
where \(r = (R_{1\, 1} + R_{1\, n-1}, R_{2\, 1} + R_{2\, n-1}, \ldots, R_{n-1\, 1} + R_{n-1\, n-1})^{\top}\). This gives
\begin{equation*}
M_1 = M_1\left(\lambda\right) = \left(\begin{array}{c|c}
m - 1 + w -\lambda  - \frac{\delta_{0} \psi}{\xi} &
\begin{array}{cccc}
-1-\frac{\delta_{1} \psi}{\xi} & -1-\frac{\delta_{2} \psi}{\xi} & \cdots & -1-\frac{\delta_{m-1} \psi}{\xi}
\end{array}
\\
\hline
-\mathbf{1}_{m-1} & \left(1-\lambda\right)I
\end{array}\right),
\end{equation*}
where \(\psi = \psi\left(w, \lambda\right)\) is given by (\ref{eq04gaigf83f42dhgfa7973}). Then, by virtue of Lemma \ref{Lem391q0kieb6fb2j2iwbjyfu5ewwe}, we have
\begin{lemma}\label{Lem54i767k57uktyhftdsi2}
For real \(\lambda > 0\), assume \(R = Q_{n-1}^{-1}\) exists, and \(\xi\left(\delta,\lambda\right) \neq 0\). Then, \(\lambda \neq 1\) is an eigenvalue of \(L_{G_{p}}\) if and only if
\begin{equation}\label{eq6875i7uytetck7rlkycfkuyfs}
\lambda^{2} - \left[m + w - \frac{\delta_{0} \psi}{\xi}\right]\lambda + w - \frac{\delta \psi}{\xi} = 0,
\end{equation}
or equivalently, one of the following holds:
\begin{equation}\label{eq02ut6rkuylwr68wrewufal}
\lambda = \frac{1}{2}\left(m + w - \sqrt{\left(m + w\right)^2 - 4w + \frac{4\left(\delta - \delta_{0}\lambda\right) \psi}{\xi}}\right)
\end{equation}
or
\begin{equation*}
\lambda = \frac{1}{2}\left(m + w + \sqrt{\left(m + w\right)^2 - 4w + \frac{4\left(\delta - \delta_{0}\lambda\right) \psi}{\xi}}\right).
\end{equation*}
\end{lemma}

\begin{remark}
If \(\lambda \notin \lbrace \beta^{-}_{m, w}, \beta^{+}_{m, w}\rbrace\), then relation (\ref{eq6875i7uytetck7rlkycfkuyfs}) can be derived from the equation \(\xi - y\psi = 0\) (see (\ref{eq6545eytd7wicq67rewq})), and vice versa. In other words, if \(\lambda\neq\beta^{\pm}_{m, w}\), then relation (\ref{eq6875i7uytetck7rlkycfkuyfs}) does not give any further information about the eigenvalue \(\lambda\) other than what \(P_{2} = 0\) gives, where \(P_{2}\) is given by (\ref{eq6545eytd7wicq67rewq}). However, since \(P_{2}\) is not defined at \(\beta^{\pm}_{m, w}\) (because \(\mu\) is not defined at these points), we still require (\ref{eq6875i7uytetck7rlkycfkuyfs}) to analyze \(\lambda = \beta^{\pm}_{m, w}\).
\end{remark}

\subsection{Proof of Theorem \ref{Theorem18na75t87ckq7tyrm3wa}}\label{Sec378uty7tdguw}

In this section, we prove Theorem \ref{Theorem18na75t87ckq7tyrm3wa}. Throughout this section, we assume that \(w_{0} = 1\) and \(w_{i} = 0\), where \(1\leq i \leq n-1\). Moreover, we have that \(\delta \geq \delta_{0}\geq 0\). Note that, to adapt this proof for the case of Theorem \ref{Thm838o8rauger}, it is sufficient to assume \(\delta = \delta_{0}\). We start with the following definition.
\begin{definition}\label{Defno8ylbywl83v6cu4rty4etr}
Recall Definition \ref{Defn75u6ru6ru452e5}. Assume \(\beta_{m, 1}^{-}\notin\{\alpha_{l}:\, 0\leq l\leq n\}\) and let \(\kappa \geq 2\) be the even integer such that \(\beta_{m, 1}^{-} \in \left(\alpha_{\kappa -2}, \alpha_{\kappa}\right)\).
\begin{enumerate}[(i)]
\item Define \(J_{\beta^{-}} := \left(\alpha_{\kappa -2}, \alpha_{\kappa}\right)\).
\item Let \(2 \leq l \leq n-2\) be even. We define
\begin{equation*}
J_{l} = \left\lbrace\begin{array}{ll}
\left(\alpha_{l-1}, \alpha_{l}\right), \quad & \mathrm{ if }\,\, 2 \leq l < \kappa,\\
\left(\alpha_{l}, \alpha_{l+1}\right), \quad & \mathrm{ if }\,\, \kappa\leq l \leq n-2.
\end{array}\right.
\end{equation*}
\item Define \(J_{\beta^{+}} := \left(\alpha_{n-1}, \infty\right)\).
\item For the sake of convenience, we define the set of indices \(\mathcal{I} := \{\beta^{-}, \beta^{+}\} \cup \{l:\, 0< l < n \mathrm{\,\, and\,\, } l \mathrm{\,\, is\,\, even}\}\).
\end{enumerate}
\end{definition}

\begin{remark}
Note that when \(\kappa = 2\), there does not exists \(J_l\) for \(2\leq l < \kappa\).
\end{remark}

\begin{remark}
Notice that \(\beta_{m,1}^{+} > m \geq 4\), and so \(\beta_{m,1}^{+} \in J_{\beta^{+}}\).
\end{remark}

Considering eigenvalues with their multiplicities, the modified Laplacian \(L_{G_{p}}\) has \(n+m\) eigenvalues. The next lemma describes where these \(n+m\) eigenvalues are located.
\begin{lemma}\label{Lem75u6urcu64u6r6ej5ywr75}
Let \(\overline{\delta}\neq 0\) be an arbitrary modification that satisfies \(\delta < \delta_{0}\beta^{+}_{m,1}\). Then, all the \(n+m\) eigenvalues of the modified Laplacian \(L_{G_p}\) of model II are real and given by the union of the following four disjoint groups (see also Remark \ref{Rem7i7qr5iuwtq7ti37t9irty}).
\begin{enumerate}[(i)]
\item\label{Item859su6tdutiyarr8eylcme} \(L_{G_{p}}\) has \(\lfloor \frac{n-1}{2}\rfloor + 1\) real eigenvalues given by \(\{\alpha_{l}: \mathrm{\,\, where \,\,} 0\leq l \leq n-1 \mathrm{\,\, and \,\,} l \mathrm{\,\, is\,\, even}\}\).
\item\label{Item637t64r3864y3fydfeq2} \(L_{G_{p}}\) has \(m-2\) of repeated eigenvalue \(\lambda = 1\).
\item\label{Item847ti7tiyftdyrgouza2re} Recall the set \(\mathcal{I}\). Each interval \(J_{\gamma}\) for \(\gamma\in\mathcal{I}\) and \(\gamma \neq \beta^{+}\) contains exactly one real eigenvalue of \(L_{G_{p}}\) (except possibly for the \(m-2\) eigenvalues \(1\) counted in item (\ref{Item637t64r3864y3fydfeq2})). We have \(\lfloor \frac{n}{2}\rfloor\) of these intervals, and so \(L_{G_{p}}\) has \(\lfloor \frac{n}{2}\rfloor\) real eigenvalues given by these intervals.
\item\label{Item86i75tyfjtqdu3oug4iuy} The interval \(J_{\beta^{+}}\) contains two real eigenvalues of the modified Laplacian \(L_{G_{p}}\). Thus, \(L_{G_{p}}\) has \(2\) eigenvalues given by \(J_{\beta^{+}}\).
\end{enumerate}
\end{lemma}

\begin{remark}
Observe that \(\left(\lfloor \frac{n-1}{2}\rfloor + 1\right) + \left(m-2\right) + \lfloor \frac{n}{2}\rfloor + 2 = n+m\).
\end{remark}

\begin{remark}\label{Rem7i7qr5iuwtq7ti37t9irty}
The sets of the eigenvalues given by items (\ref{Item859su6tdutiyarr8eylcme}) and (\ref{Item637t64r3864y3fydfeq2}) might not be disjoint, i.e. \(\alpha_{l} = 1\) for some even \(l\). The same may happen for (\ref{Item637t64r3864y3fydfeq2}) and (\ref{Item847ti7tiyftdyrgouza2re}), i.e. the eigenvalue in \(J_{\gamma}\) given by item (\ref{Item847ti7tiyftdyrgouza2re}) equals to \(1\). The eigenvalue \(1\) in such scenarios are counted separately from the \(m-2\) eigenvalues \(1\) given in item (\ref{Item637t64r3864y3fydfeq2}). In such scenarios, the multiplicity of eigenvalue \(1\) is \(m-1\).
\end{remark}

The proof of Lemma \ref{Lem75u6urcu64u6r6ej5ywr75} is postponed to Section \ref{Sec53452y75u46rjy6e52jd2jy}. We now prove Theorem \ref{Theorem18na75t87ckq7tyrm3wa}. 

Part (\ref{Item46gktcytrxjt00}) of Theorem \ref{Theorem18na75t87ckq7tyrm3wa} follows from Theorem \ref{Thm274i7tigfdyf} which is proved later in Section \ref{Sec12384o8y8agrit4}. Here, we show that part (\ref{Item46gktcytrxjt00}) of Theorem \ref{Theorem18na75t87ckq7tyrm3wa} satisfies the corresponding assumptions of Theorem \ref{Thm274i7tigfdyf}. Recall \(S\) given by (\ref{eq285tguiuugftabigsipewsf}). Setting \(w_{0} = 1\) and \(w_{i} = 0\), where \(1\leq i \leq n-1\), gives
\begin{equation}
S = \frac{\alpha_{2}\left[\delta \left(\alpha_{2} - m\right) + \delta_{0} - \delta\right]}{\left(\alpha_{2} -\beta^{-}_{m,1}\right) \left(\alpha_{2} -\beta^{+}_{m,1}\right)}
\end{equation}
and
\begin{equation}
\sum_{i=0}^{n-1} w_{i}\cos\left(\frac{n}{2} - i\right)\theta = \cos\frac{n\theta}{2}, \qquad \mathrm{where}\qquad \theta = \pi - \cos^{-1}\left(\frac{\beta^{-}_{m,1} -2}{2}\right).
\end{equation}
Take into account that \(\delta_{0} \leq \delta\) and \(\alpha_{2} < 4\). When \(\alpha_{2} < \beta^{-}_{m,1}\), we have \(S < 0\). On the other hand, when \(\alpha_{1} < \beta^{-}_{m,1} < \alpha_{2}\), we have \(\frac{\pi}{n}< \theta < \frac{2\pi}{n}\) and so \(\cos\frac{n\theta}{2} < 0\), where \(\theta\) is as above. Therefore, part (\ref{Item46gktcytrxjt11}) of Theorem \ref{Theorem18na75t87ckq7tyrm3wa} follows from parts (\ref{Item42iyo8yaor4gtwf1}) and (\ref{Item75988tugqyefriyer333}) of Theorem \ref{Thm274i7tigfdyf}, and part (\ref{Item46gktcytrxjt33}) of Theorem \ref{Theorem18na75t87ckq7tyrm3wa} follows directly from part (\ref{Item75988tugqyefriyer111}) of Theorem \ref{Thm274i7tigfdyf}.

Part (\ref{Item46gktcytrxjt55}) of Theorem \ref{Theorem18na75t87ckq7tyrm3wa} is a consequence of part (\ref{Item538youkugaliyfyejr11}) of Theorem \ref{Thm274i7tigfdyf}, since, as mentioned above, when \(\alpha_{2} < \beta^{-}_{m,1}\), we have \(S < 0\).

Let us now prove part (\ref{Item46gktcytrxjt22}) of Theorem \ref{Theorem18na75t87ckq7tyrm3wa}. First, assume \(\alpha_{2} < \beta^{-}_{m,1}\). This implies \(\kappa > 2\). Thus, by Lemma \ref{Lem75u6urcu64u6r6ej5ywr75}, \(L_{G_{p}}\) has a unique eigenvalue in the interval \(J_{2} = \left(\alpha_{1}, \alpha_{2}\right)\) which is indeed the spectral gap of \(L_{G_{p}}\). Denote it by \(\lambda_{2}\left(L_{G_p}\right)\). Since the spectral gap of the unmodified Laplacian \(L_{G}\) is \(\alpha_{2}\), we have \(\lambda_{2}\left(L_{G_p}\right) < \lambda_{2}\left(L_{G}\right)\). This shows that, in the case \(\alpha_{2} < \beta^{-}_{m,1}\), the statement of part (\ref{Item46gktcytrxjt11}) of Theorem \ref{Theorem18na75t87ckq7tyrm3wa} holds for arbitrary modification \(\overline{\delta}\) that satisfies \(\delta < \delta_{0} \beta^{+}_{m,1}\).

Now, assume \(\beta^{-}_{m,1} < \alpha_{2}\). This implies \(\kappa = 2\), i.e. \(J_{\beta^{-}} = \left(0, \alpha_{2}\right)\). According to Lemma \ref{Lem75u6urcu64u6r6ej5ywr75}, \(L_{G_{p}}\) has a unique eigenvalue in the interval \(J_{\beta^{-}} = \left(0, \alpha_{2}\right)\) which is indeed the spectral gap of \(L_{G_{p}}\). Denote it by \(\lambda_{2}\left(\overline{\delta}\right)\). Note that the spectral gap of the unmodified graph \(L_{G}\) is \(\lambda_{2}\left(0\right) = \beta^{-}_{m,1}\). Following Lemma \ref{Lem54i767k57uktyhftdsi2}, we have
\begin{equation}\label{eq764y6rytdytdjetwq}
\lambda_{2}\left(\overline{\delta}\right) = \frac{1}{2}\left(m + 1 - \sqrt{\left(m + 1\right)^2 - 4 + \frac{4\left[\delta - \delta_{0}\lambda_{2}\left(\overline{\delta}\right)\right]}{\xi}}\right),
\end{equation}
where \(\xi = \xi\left(\delta, \lambda_{2}\left(\overline{\delta}\right)\right) = \delta -2\sin\theta \tan\frac{n\theta}{2}\) and \(\theta = \pi - \cos^{-1}(\frac{\lambda_{2}\left(\overline{\delta}\right) -2}{2})\). Observe that when \(\overline{\delta} = 0\) (and consequently, \(\delta = \delta_{0} = 0\)), \(\lambda_{2}\left(0\right) = \beta^{-}_{m,1}\) satisfies this relation. According to (\ref{eq764y6rytdytdjetwq}), the proof follows from this observation that for a given \(\overline{\delta}\), we have \(\lambda_{2}\left(\overline{\delta}\right) > \lambda_{2}\left(0\right)\), \(\lambda_{2}\left(\overline{\delta}\right) = \lambda_{2}\left(0\right)\) and \(\lambda_{2}\left(\overline{\delta}\right) < \lambda_{2}\left(0\right)\) if and only if the expression
\begin{equation}\label{eq157tqut7tr3rg4jrye}
\frac{\delta - \delta_{0}\lambda_{2}\left(\overline{\delta}\right)}{\xi}
\end{equation}
be negative, zero and positive, respectively.

First, consider the case \(\alpha_{1} < \lambda_{2}\left(0\right) = \beta^{-}_{m,1} < \alpha_{2}\). It is easily seen that \(\beta^{-}_{m,1} \leq \beta^{-}_{4,1} \approx 0.21\) for all \(m\geq 4\). Thus, having \(\alpha_{1} = 2\left(1-\cos\frac{\pi}{n}\right) < \beta^{-}_{m, 1} < 0.21\) yields \(n\geq 7\) which implies \(\alpha_{2} = 2\left(1-\cos\frac{2\pi}{n}\right) < 1\). Note also that as \(\overline{\delta}\) changes, \(\lambda_{2}\left(\overline{\delta}\right)\) remains in \(\left(0, \alpha_{2}\right)\) (this is a consequence of part (\ref{Item847ti7tiyftdyrgouza2re}) of Lemma \ref{Lem75u6urcu64u6r6ej5ywr75}). Therefore
\begin{equation}\label{eq08973i74tu36rde3r}
\delta - \delta_{0}\lambda_{2}\left(\overline{\delta}\right) > \delta - \delta_{0} + \delta_{0}\left[1-\alpha_{2}\right] \geq \max\{\delta - \delta_{0}, \delta_{0}\left[1-\alpha_{2}\right]\}.
\end{equation}
Thus, the numerator of (\ref{eq157tqut7tr3rg4jrye}) is positive for any \(\overline{\delta} \neq 0\). Regarding the denominator of (\ref{eq157tqut7tr3rg4jrye}), note that \(\xi(0, \beta^{-}_{m,1}) > 0\) when \(\alpha_{1} < \beta^{-}_{m,1} < \alpha_{2}\). We claim that \(\xi\left(\delta, \lambda_{2}\left(\overline{\delta}\right)\right) > 0\) for all \(\overline{\delta}\). Taking into account that \(\xi\) is a smooth function of \(\left(\delta, \lambda\right)\) for \(\lambda \neq \alpha_{1}\), the claim will be proved once we show that \(\lambda_{2}(\overline{\delta}) > \alpha_{1}\) holds for any \(\overline{\delta}\) and also \(\xi\) does not vanish as \(\overline{\delta}\) varies.

We first show that \(\lambda_{2}(\overline{\delta}) > \alpha_{1}\) for all \(\overline{\delta}\). Assume the contrary; there exists \(\overline{\delta}^{\dagger}\) and correspondingly \(\delta^{\dagger}\) and \(\delta_{0}^{\dagger}\) for which \(\lambda_{2}(\overline{\delta}^{\dagger}) = \alpha_{1}\). However, \(\lim_{\delta\rightarrow \delta^{\dagger}} \xi = \infty\). On the other hand, the numerator of (\ref{eq157tqut7tr3rg4jrye}) converges to \(\delta^{\dagger} - \delta_{0}^{\dagger}\alpha_{1} > 0\). Therefore, as \(\delta\rightarrow \delta^{\dagger}\), expression (\ref{eq157tqut7tr3rg4jrye}) converges to zero which, by (\ref{eq764y6rytdytdjetwq}), implies that \(\lambda_{2}(\overline{\delta}^{\dagger}) = \beta^{-}_{m,1}\) and so \(\beta^{-}_{m,1} = \alpha_{1}\). This contradicts the assumption \(\beta_{m, 1}^{-}\notin\{\alpha_{l}:\, 0\leq l<n\}\) of Theorem \ref{Theorem18na75t87ckq7tyrm3wa}. Thus, \(\lambda_{2}(\overline{\delta}) < \alpha_{1}\) for all \(\overline{\delta}\).

Since \(\alpha_{1} < \lambda_{2}(\overline{\delta}) < \alpha_{2}\) for all \(\overline{\delta}\), we have that \(\sin\theta \tan\frac{n\theta}{2} < 0\), where \(\theta = \pi - \cos^{-1}(\frac{\lambda_{2}\left(\overline{\delta}\right) -2}{2})\). This yields \(\xi > \delta\) for all \(\overline{\delta}\) which means that it cannot vanish as \(\overline{\delta}\) varies. Therefore, the numerator and denominator of  (\ref{eq157tqut7tr3rg4jrye}) are positive. It then follows from (\ref{eq08973i74tu36rde3r}), that when \(\alpha_{1} < \beta^{-}_{m,1} < \alpha_{2}\), we have \(\lambda_{2}(L_{G_p}) < \lambda_{2}(L_{G})\), as desired.

Now, we consider the case \(0 < \lambda_{2}\left(0\right) = \beta^{-}_{m,1} < \alpha_{1}\). We first show that \(\lambda_{2}(\overline{\delta}) < \alpha_{1}\) for all \(\overline{\delta}\). Assume the contrary; there exists \(\overline{\delta}^{\dagger}\) and correspondingly \(\delta^{\dagger}\) and \(\delta_{0}^{\dagger}\) for which \(\lambda_{2}(\overline{\delta}^{\dagger}) = \alpha_{1}\). However, \(\lim_{\delta\rightarrow \delta^{\dagger}} \xi = -\infty\). On the other hand the numerator of (\ref{eq157tqut7tr3rg4jrye}) converges to \(\delta^{\dagger} - \delta_{0}^{\dagger}\alpha_{1} \geq 0\) (note that \(\alpha_{1}\leq 1\) for all \(n\geq 3\)). Therefore, as \(\delta\rightarrow \delta^{\dagger}\), expression (\ref{eq157tqut7tr3rg4jrye}) converges to zero which, by (\ref{eq764y6rytdytdjetwq}), implies that \(\lambda_{2}(\overline{\delta}^{\dagger}) = \beta^{-}_{m,1}\) and so \(\beta^{-}_{m,1} = \alpha_{1}\). This contradicts the assumption \(\beta_{m, 1}^{-}\notin\{\alpha_{l}:\, 0\leq l<n\}\) of Theorem \ref{Theorem18na75t87ckq7tyrm3wa}. Thus, \(\lambda_{2}(\overline{\delta}) < \alpha_{1}\) for all \(\overline{\delta}\).

It is easily seen that \(\xi(0, \beta^{-}_{m,1}) < 0\) when \(0 < \beta^{-}_{m,1} < \alpha_{1}\). We claim that \(\xi(0, \beta^{-}_{m,1}) < 0\) for all \(\overline{\delta}\). Note that \(\xi\) is a smooth function for \(\lambda \neq \alpha_{1}\). On the other hand, we have shown that \(\lambda_{2}(\overline{\delta}) < \alpha_{1}\) for all \(\overline{\delta}\). Thus, to prove the claim, we need to show that \(\xi\) does not vanish as \(\overline{\delta}\) varies. Assume the contrary; there exists \(\overline{\delta}^{\dagger}\) and correspondingly \(\delta^{\dagger}\) and \(\delta_{0}^{\dagger}\) such that as \(\overline{\delta} \rightarrow \overline{\delta}^{\dagger}\), we have \(\xi (\delta, \lambda_{2}(\overline{\delta}^{\dagger})) \rightarrow 0\). By (\ref{eq764y6rytdytdjetwq}), this requires the numerator of (\ref{eq157tqut7tr3rg4jrye}) to vanish at \(\overline{\delta}^{\dagger}\), i.e. \(\delta^{\dagger} - \delta_{0}^{\dagger}\lambda_{2}(\overline{\delta}^{\dagger}) = 0\). However, by \(\lambda_{2}(\overline{\delta}) < \alpha_{1}\) and taking into account that \(\alpha_{1} \leq 1\) for all \(n\geq 3\), we obtain
\begin{equation}
\delta^{\dagger} - \delta_{0}^{\dagger}\lambda_{2}\left(\overline{\delta}^{\dagger}\right) \geq \max\left\lbrace\delta^{\dagger} - \delta_{0}^{\dagger}, \delta_{0}^{\dagger}\left[1-\lambda_{2}\left(\overline{\delta}^{\dagger}\right)\right]\right\rbrace > 0.
\end{equation}
This contradicts the assumption of vanishing \(\xi\) at \(\overline{\delta}^{\dagger}\). Therefore, we have \(\xi(0, \beta^{-}_{m,1}) < 0\) for all \(\overline{\delta}\). It then follows from (\ref{eq08973i74tu36rde3r}), that when \(0 < \beta^{-}_{m,1} < \alpha_{1}\), we have \(\lambda_{2}(L_{G_p}) > \lambda_{2}(L_{G})\), as desired. This finishes the proof of part (\ref{Item46gktcytrxjt22}) and the proof of Theorem \ref{Theorem18na75t87ckq7tyrm3wa}.

\subsubsection{Proof of Lemma \ref{Lem75u6urcu64u6r6ej5ywr75}}\label{Sec53452y75u46rjy6e52jd2jy}

So far, we have used Lemma \ref{Lem75u6urcu64u6r6ej5ywr75} to prove Theorem \ref{Theorem18na75t87ckq7tyrm3wa}. We are now in the position of proving this lemma.

The proof of Lemma \ref{Lem75u6urcu64u6r6ej5ywr75} is based on Lemma \ref{Lem526ri76er6cp95c2l7xfwsw1}. In the setting of Theorems \ref{Thm838o8rauger} and \ref{Theorem18na75t87ckq7tyrm3wa}, we assume \(w_{0} = 1\) and \(w_{i} = 0\) for \(i=1,\ldots, n-1\). In this case, (\ref{eq7guyf7t2fuyf6f2ut23sh}) is written as
\begin{equation}\label{eq42yf6riyugjyfhtdqw2321hfyts}
p\left(\theta\right) = 2\left[\cos n\theta - 1\right] + \left[\delta - y\right]\cdot\frac{\sin n\theta}{\sin\theta}.
\end{equation}
Then, Lemma \ref{Lem526ri76er6cp95c2l7xfwsw1} gives
\begin{lemma}\label{Lem92i755ku75kuy37it454y}
For \(0 \leq \lambda < 4\), let \(\theta = \pi - \cos^{-1}(\frac{\lambda -2}{2})\). Consider \(p\) given by (\ref{eq42yf6riyugjyfhtdqw2321hfyts}). Then
\begin{enumerate}[(i)]
\item\label{Itemi65quku64k734yk3i9} for even \(0 \leq l \leq n-1\), we have \(p(\frac{l\pi}{n}) = 0\).
\item for odd \(1 \leq l \leq n-1\), we have \(p(\frac{l\pi}{n}) = -4\).
\item for even \(1 \leq l \leq n-1\), we have
\begin{equation}\label{eq65teu0ht5cr5etdu6rw}
p^{\prime}\left(\frac{l\pi}{n}\right) = \frac{n}{\sin \frac{l\pi}{n}}\left[\delta - y\right] = \frac{-n\lambda}{\sin\frac{l\pi}{n}}\cdot \frac{\left(m-\lambda\right)\delta + \delta - \delta_{0}}{\lambda^{2} - \left(m+1\right)\lambda + 1}.
\end{equation}
\end{enumerate}
\end{lemma}

\begin{proof}[Proof of part (\ref{Item859su6tdutiyarr8eylcme}) of Lemma \ref{Lem75u6urcu64u6r6ej5ywr75}]
The proof directly follows from part (\ref{Itemi65quku64k734yk3i9}) of Lemma \ref{Lem92i755ku75kuy37it454y}.
\end{proof}

\begin{proof}[Proof of part (\ref{Item637t64r3864y3fydfeq2}) of Lemma \ref{Lem75u6urcu64u6r6ej5ywr75}]
The proof directly follows from Lemma \ref{Lem36iyguyfhctgczhvjgwi} and its proof.
\end{proof}

\begin{proof}[Proof of part (\ref{Item847ti7tiyftdyrgouza2re}) of Lemma \ref{Lem75u6urcu64u6r6ej5ywr75}]
We first investigate \(p^{\prime}\left(\frac{l\pi}{n}\right)\) given by (\ref{eq65teu0ht5cr5etdu6rw}). The expression \(\lambda^{2} - \left(m+1\right)\lambda + 1\) is positive if and only if \(\lambda < \beta_{m,1}^{-}\) or \(\lambda > \beta_{m,1}^{+} > 4\). For \(\lambda = 2[1-\cos \frac{l\pi}{n}]\), when \(l\) is even, this gives
\begin{equation*}
\left\lbrace\begin{array}{ll}
\lambda^{2} - \left(m+1\right)\lambda + 1 > 0, \quad & \mathrm{ if\,\, } l \mathrm{\,\, is\,\, even\,\, and \,\, } 2 \leq l < \kappa,\\
\lambda^{2} - \left(m+1\right)\lambda + 1 < 0, \quad & \mathrm{ if\,\, } l \mathrm{\,\, is \,\, even \,\, and \,\,} \kappa \leq l < n-1.
\end{array}\right.
\end{equation*}

When \(\lambda \in J_{\gamma}\), for \(\gamma \neq \beta^{+}\), we have that \(\lambda < 4 \leq m\). On the other hand, \(\delta_{0}\leq \delta\). This implies that when \(\delta > 0\), we have \(\left(m-\lambda\right)\delta + \delta - \delta_{0} > 0\). Taking into account that \(\sin\frac{l\pi}{n}>0\) for all \(0\leq l\leq n-1\), we obtain
\begin{equation*}
\left\lbrace\begin{array}{ll}
p^{\prime}\left(\frac{l\pi}{n}\right) < 0, \quad & \mathrm{ if \,\,} l \mathrm{\,\, is \,\, even \,\, and \,\,} 2 \leq l < \kappa,\\
p^{\prime}\left(\frac{l\pi}{n}\right) > 0, \quad & \mathrm{ if \,\,} l \mathrm{\,\, is \,\, even \,\, and \,\,} \kappa \leq l < n-1.
\end{array}\right.
\end{equation*}
We have that \(p^{\prime}\left(\frac{l\pi}{n}\right) = 0\) if and only if \(\delta = 0\). This, together with item (\ref{Itemi65quku64k734yk3i9}) of Lemma \ref{Lem92i755ku75kuy37it454y}, gives
\begin{proposition}
Any point of \(\{\alpha_{l}:\,\mathrm{ where \,\,} 0\leq l \leq n-1 \mathrm{\,\, and \,\,} l \mathrm{\,\, is \,\, even}\}\) is a multiple eigenvalue of \(L_{G}\) with multiplicity \(2\), and a simple eigenvalue of \(L_{G_{p}}\).
\end{proposition}
For even \(l\), when \(2 \leq l < \kappa\), we have \(p^{\prime}\left(\frac{l\pi}{n}\right) < 0\). This means that \(p\left(\theta\right)\) is positive for \(\theta\) close to \(\frac{l\pi}{n}\) and \(\theta < \frac{l\pi}{n}\). On the other hand, \(p\left(\frac{\left(l-1\right)\pi}{n}\right) = -4 < 0\). Thus, by the intermediate value theorem, the function \(p\) has a root in the interval \(\left(\frac{\left(l-1\right)\pi}{n}, \frac{l\pi}{n}\right)\). This implies that \(L_{G_{p}}\) has a real eigenvalue in \(J_{l}\). Analogously, for even \(l\) and when \(\kappa \leq l < n-1\), the function \(p\) has a root in the interval \(\left(\frac{l\pi}{n}, \frac{\left(l+1\right)\pi}{n}\right)\) which means that \(L_{G_{p}}\) has a real eigenvalue in \(J_{l}\).

At \(\overline{\delta} = 0\) (when there is no modification), the interval \(J_{\beta^{-}}\) has the eigenvalue \(\beta^{-}_{m,1}\). As \(\overline{\delta}\) changes, the eigenvalue \(\beta^{-}_{m,1}\) starts to move. However, since \(\alpha_{\kappa -2}\) and \(\alpha_{\kappa}\) are simple roots, this eigenvalue cannot leave the interval \(J_{\beta^{-}} = \left(\alpha_{\kappa -2}, \alpha_{\kappa}\right)\). This means that \(L_{G_{p}}\) has a real root in \(J_{\beta^{-}}\).

We have shown that each interval \(J_{\gamma}\) for \(\gamma\in\mathcal{I}\) and \(\gamma \neq \beta^{+}\) contains at least one real eigenvalue of \(L_{G_{p}}\). To finish the proof, we need to show that each of these intervals has exactly one eigenvalue (apart from \(m-2\) eigenvalues \(1\) counted in item (\ref{Item637t64r3864y3fydfeq2}) that might be located in one of these intervals). Note that, we have already counted \(n+m-2 = \lfloor \frac{n-1}{2}\rfloor + 1 + m-2 + \lfloor \frac{n}{2}\rfloor\) real eigenvalues of \(L_{G_{p}}\). The matrix \(L_{G_{p}}\) has \(n+m\) eigenvalues. Thus, the proof of part (\ref{Item847ti7tiyftdyrgouza2re}) of this lemma is done after we prove part (\ref{Item86i75tyfjtqdu3oug4iuy}) of this lemma below.
\end{proof}

\begin{proof}[Proof of part (\ref{Item86i75tyfjtqdu3oug4iuy}) of Lemma \ref{Lem75u6urcu64u6r6ej5ywr75}]
So far, we have shown that the matrix \(L_{G_{p}}\) has at least \(n+m-2\) eigenvalues located outside of the interval \(J_{\beta^{+}}\), and as \(\overline{\delta}\) varies, none of these eigenvalues enters this interval. Notice that for arbitrary \(\delta > 0\), when \(n\) is even, \(p(\frac{\left(n-1\right)\pi}{n}) < 0\), and when \(n\) is odd, \(p(\frac{\left(n-1\right)\pi}{n}) = 0\) and \(p^{\prime}(\frac{\left(n-1\right)\pi}{n}) \neq 0\). This means that if there is any real eigenvalue located in \(J_{\beta^{+}}\), then it cannot leave this interval as \(\overline{\delta}\) changes. As shown below, for sufficiently small \(\overline{\delta}\neq 0\), the interval \(J_{\beta^{+}}\) has exactly two real eigenvalues. On the other hand, \(L_{G_p}\) is a real matrix. Therefore, if it possesses non-real eigenvalues, then they need to appear as pairs (complex conjugates). This means that, for a given \(\overline{\delta}\), we either have two real eigenvalues in \(J_{\beta^{+}}\) or none. Therefore, the proof of part (\ref{Item86i75tyfjtqdu3oug4iuy}) of Lemma \ref{Lem75u6urcu64u6r6ej5ywr75} is done if we show that, under the condition \(\delta < \delta_{0}\beta_{m,1}^{+}\), the interval \(J_{\beta^{+}}\) has at least one real eigenvalue.

First, we show that when \(\overline{\delta}\neq 0\) is sufficiently small, the interval \(J_{\beta^{+}}\) has exactly two real eigenvalues. For even \(n\), this is obvious since at \(\delta =0\), we have two eigenvalues \(\lambda = 4\) and \(\lambda = \beta^{+}_{m,1}\), and so, as \(\overline{\delta}\) varies and remains sufficiently small, these two eigenvalues might move but they remain in \(J_{\beta^{+}}\) and do not collide (so, they remain real). The case of odd \(n\) is similar; since for \(\overline{\delta}\neq 0\), we have \(p^{\prime}\left(\alpha_{n-1}\right) > 0\) and \(p(4) < 0\), the intermediate value theorem implies that there is a root in the interval \((\alpha_{n-1}, 4)\). On the other hand, the eigenvalue \(\beta^{+}_{m,1}\in (4, \infty)\) of \(L_{G}\) might move as \(\overline{\delta}\) varies but as far as \(\overline{\delta}\) is sufficiently small, it does not collied with the eigenvalue that we just found in the interval \((\alpha_{n-1}, 4)\). Therefore, we have that for small \(\overline{\delta}\neq 0\), the interval \(J_{\beta^{+}}\) contains exactly two real eigenvalues of \(L_{G_{p}}\).

We now prove that \(J_{\beta^{+}}\) has at least one real root when \(\delta < \delta_{0}\beta_{m,1}^{+}\). Evaluating (\ref{eq6545eytd7wicq67rewq}) at \(\lambda > 4\), gives
\begin{equation*}
P_{2}\left(\lambda\right) = \left(-1\right)^{n-1} \cdot \frac{\sinh n\theta}{\sinh\theta} \left[\frac{2\sinh\theta}{\sinh n\theta}\cdot\left[\left(-1\right)^{n} - \cosh n\theta\right] + \delta - \frac{\delta - \delta_{0} \lambda}{\lambda^{2} - \left(m+1\right)\lambda + 1}\right],
\end{equation*}
where \(\theta = \cosh^{-1}(\frac{\lambda -2}{2})\). Note that \(\lambda^{2} - \left(m+1\right)\lambda + 1\) vanishes at \(\lambda = \beta^{\pm}_{m,1}\). Thus, when \(\delta < \delta_{0}\beta_{m,1}^{+}\), we have
\begin{equation}\label{eq7y97g7g484gt8g58tg5gt}
\lim P_{2}\left(\lambda\right) = \left\lbrace\begin{array}{ll}
+\infty, \quad & \mathrm{\,\,for \,\, even \,\,} n, \mathrm{\,\, as \,\,} \lambda\rightarrow {\left(\beta_{m,1}^{+}\right)}^{-},\\
-\infty, \quad & \mathrm{\,\,for \,\, odd \,\,} n, \mathrm{\,\, as \,\,} \lambda\rightarrow {\left(\beta_{m,1}^{+}\right)}^{-}.
\end{array}\right.
\end{equation}
When \(n\) is even, we have \(P_{2}\left(4\right) = \frac{4n\left[\left(m-3\right)\delta - \delta_{0}\right]}{13-4m} < 0\). Taking (\ref{eq7y97g7g484gt8g58tg5gt}) and the fact that \(P_{2}\) is smooth on \((\alpha_{n-1}, \beta_{m,1}^{+})\) into account (see Lemma \ref{Lemy9yq7g3ryg4frewiyjfre}), the intermediate value theorem implies the existence of a real root of \(P_{2}\) in \((4, \beta_{m,1}^{+}) \subset J_{\beta^{+}}\), as desired.

For the case of odd \(n\), we have \(p(\frac{\left(n-1\right)\pi}{n}) = 0\) and \(p^{\prime}(\frac{\left(n-1\right)\pi}{n}) > 0\). Thus, for \(\lambda > \alpha_{n-1}\) and close to \(\alpha_{n-1}\), we have \(P_{2}(\lambda) > 0\). Taking (\ref{eq7y97g7g484gt8g58tg5gt}) and the fact that \(P_{2}\) is smooth on \((\alpha_{n-1}, \beta_{m,1}^{+})\) into account (see Lemma \ref{Lemy9yq7g3ryg4frewiyjfre}), the intermediate value theorem implies the existence of a real root of \(P_{2}\) in \((\alpha_{n-1}, \beta_{m,1}^{+}) \subset J_{\beta^{+}}\), as desired. This ends the proof.
\end{proof}

\subsection{Proof of Theorem \ref{Thm274i7tigfdyf}}\label{Sec12384o8y8agrit4}

%\subsubsection{Proof of part (\ref{Item456rcuarl6kuwer}) of Theorem \ref{Thm274i7tigfdyf}}

We first prove that for sufficiently small \(\overline{\delta}\), all the eigenvalues of \(L_{G_p}\) are real. It is known that the roots of a polynomial (in our case, the characteristic polynomial of \(L_{G_{p}}\)) depend continuously on the coefficients of that polynomial. Therefore, if \(\{\lambda_{i}\left(\overline{\delta}\right):\, i=1,\ldots,n+m\}\) is the spectrum of \(L_{G_{p}}\), then \(\lambda_{i}\left(\overline{\delta}\right)\) is a continuous function of \(\overline{\delta}\). It is a direct consequence of the implicit function theorem that if \(\lambda_{i}\left(0\right)\) is a simple eigenvalue of \(L_{G}\), then for sufficiently small \(\overline{\delta}\), we have that \(\lambda_{i}\left(\overline{\delta}\right)\) is real. Thus, to prove our statement, we need to investigate how multiple eigenvalues of the unmodified Laplacian \(L_{G}\) behave as \(\overline{\delta}\) varies.

Recall Proposition \ref{Prop87986kiuglru3aw3er}. According to Remark \ref{Rem95jyru6rmtdhtlirstw} and the assumption \(\beta_{m, w}^{-}\notin\{\alpha_{l}:\, 0\leq l<n\}\) of the theorem, we have that \(\beta^{-}_{m, w}\) and \(\beta^{+}_{m, w}\) are simple eigenvalues of \(L_{G}\). Note that \(1\in \{\alpha_{l}:\, 0\leq l\leq n, \mathrm{\,\, and \,\,} l \mathrm{\,\, is \,\, even}\}\) if and only if \(\frac{n}{6}\) is an integer. First, assume \(\frac{n}{6}\notin\mathbb{Z}\). In this case, the multiplicity of all the eigenvalues \(\alpha_{l}\) except for \(0\) and \(4\) (the eigenvalue \(4\) appears only when \(n\) is even) is \(2\). However, it follows from Lemma \ref{Lem526ri76er6cp95c2l7xfwsw1} that, for each even \(l\), as \(\overline{\delta}\) varies, one of the two eigenvalues \(\alpha_{l}\) remains as an eigenvalue of \(L_{G_{p}}\) for small arbitrary \(\overline{\delta}\), and the other eigenvalue moves continuously. This means that from each of the multiple eigenvalues \(\alpha_{l}\), two real eigenvalues get born. On the other hand, following Lemma \ref{Lem36iyguyfhctgczhvjgwi} and its proof, the eigenvalue \(1\) remains an eigenvalue of \(L_{G_{p}}\) with multiplicity \(m-2\). This implies that when \(\frac{n}{6} \notin \mathbb{Z}\) and \(\overline{\delta}\) is sufficiently small, all the eigenvalues of \(L_{G_{p}}\) are real. The case of \(\frac{n}{6} \in \mathbb{Z}\) is similar. With the same conclusion, except for \(l = \frac{n}{6}\), i.e. \(\alpha_{l} = 1\), two real eigenvalues get born from each eigenvalue \(\alpha_{l}\), where \(l = 1,\ldots, n-1\). Regarding \(\alpha_{l} = 1\) (note that the multiplicity of 1 as an eigenvalue of \(L_G\) in this case is \(m\)), we have that \(\alpha_{l} = 1\) remains an eigenvalue of \(L_{G_{p}}\) with multiplicity \(m-1\) and a new real eigenvalue gets born from it. This proves that when \(\overline{\delta}\) is sufficiently small, all the eigenvalues of \(L_{G_{p}}\) are real.

The rest of the proof of part (\ref{Item456rcuarl6kuwer}) of Theorem \ref{Thm274i7tigfdyf} is based on the following lemma
\begin{lemma}\label{Lem50hugiughbaujgrgadf}
Consider \(p\) and \(S\) given by (\ref{eq7guyf7t2fuyf6f2ut23sh}) and (\ref{eq285tguiuugftabigsipewsf}), respectively. We have
\begin{enumerate}[(i)]
\item \label{Item84eeur938owyre} The expressions \(S\) and \(p^{\prime}\left(\frac{2\pi}{n}\right)\) have the same sign.
\item Assume \(\alpha_{2} < \beta^{-}_{m, w}\). Then, for any arbitrary \(\overline{\delta}\), we have \(p\left(\frac{\pi}{n}\right) < 0\).
\item For sufficiently small \(\overline{\delta}\), we have \(p\left(\frac{3\pi}{n}\right) < 0\).
\end{enumerate}
\end{lemma}

\begin{proof}
The first part follows from the relation \(p^{\prime}\left(\frac{2\pi}{n}\right) = \frac{n}{\sin\frac{2\pi}{n}} S\) (see relation (\ref{eq800ug4tht9q438tqherug})). For the other two parts, note that by Lemma \ref{Lem526ri76er6cp95c2l7xfwsw1} and for odd \(1 \leq l \leq n-1\), we have \(p(\frac{l\pi}{n}) = -4 - \frac{2y}{\sin\frac{l\pi}{n}}\sum_{i=0}^{n-1} w_{i} \sin\frac{il\pi}{n}\). Regarding the case \(l=1\), assumption \(\alpha_{2} < \beta^{-}_{m, w}\) implies that \(\alpha_{2} < 1\) (see Remark \ref{Rem95jyru6rmtdhtlirstw}) and therefore \(\alpha_{1} < 1\). Thus, \(\delta - \delta_{0} \alpha_{1} > 0\), and therefore, \(y = y(\alpha_{1}) = \frac{\delta - \delta_{0} \alpha_{1}}{\alpha_{1}^{2} - \left(m+w\right)\alpha_{1} + w} > 0\). On the other hand, \(\sin \frac{i\pi}{n} \geq 0\) for \(i=0,\ldots, n-1\) and so \(\sum_{i=0}^{n-1} w_{i} \sin\frac{il\pi}{n} \geq 0\). This implies \(p\left(\frac{\pi}{n}\right) < 0\) for any \(\overline{\delta}\).

The proof of the last part follows from \(\left\vert\frac{2y}{\sin\frac{3\pi}{n}}\sum_{i=0}^{n-1} w_{i} \sin\frac{3i\pi}{n}\right\vert \ll 4\) which holds when \(\overline{\delta}\) is small enough.
\end{proof}

\begin{proof}[Proof of parts (\ref{Item42iyo8yaor4gtwf1}) and (\ref{Item42iyo8yaor4gtwf2}) of Theorem \ref{Thm274i7tigfdyf}]
Since \(\alpha_{2} < \beta^{-}_{m,w}\), the spectral gap of \(L_G\) is \(\alpha_{2}\). It follows from Lemma \ref{Lem526ri76er6cp95c2l7xfwsw1} that \(p\left(\frac{2\pi}{n}\right) = 0\) for all \(\overline{\delta}\). On the other hand, \(p\left(\frac{\pi}{n}\right)\) and \(p\left(\frac{3\pi}{n}\right)\) are both negative for sufficiently small \(\overline{\delta}\). Therefore, by intermediate value theorem, the eigenvalue that gets born from \(\alpha_{2}\) as \(\overline{\delta}\) varies is located in the interval \(\left(\alpha_{2}, \alpha_{3}\right)\) if \(p^{\prime}\left(\frac{2\pi}{n}\right) > 0\), and is located in \(\left(\alpha_{1}, \alpha_{2}\right)\) if \(p^{\prime}\left(\frac{2\pi}{n}\right) < 0\). On the other hand, by part (\ref{Item84eeur938owyre}) of Lemma \ref{Lem50hugiughbaujgrgadf}, we have that \(p^{\prime}\left(\frac{2\pi}{n}\right)\) and \(S\) have the same sign. This proves parts (\ref{Item42iyo8yaor4gtwf1}) and (\ref{Item42iyo8yaor4gtwf2}) of Theorem \ref{Thm274i7tigfdyf}.
\end{proof}

\begin{proof}[Proof of parts (\ref{Item75988tugqyefriyer111}), (\ref{Item75988tugqyefriyer222}) and (\ref{Item75988tugqyefriyer333}) of Theorem \ref{Thm274i7tigfdyf}]
Since \(\beta^{-}_{m,w} < \alpha_{2}\), the spectral gap of \(L_G\) is \(\beta^{-}_{m,w}\). So, we need to see how \(\beta^{-}_{m,w}\left(\overline{\delta}\right)\) changes as \(\overline{\delta}\) varies. By (\ref{eq02ut6rkuylwr68wrewufal}), for sufficiently small \(\overline{\delta} \neq 0\), we have that \(\beta^{-}_{m,w}\left(\overline{\delta}\right) > \beta^{-}_{m,w}\) if \(\frac{\left(\delta - \delta_{0} \beta^{-}_{m,w}\right) \psi}{\xi} < 0\), and \(\beta^{-}_{m,w}\left(\overline{\delta}\right) < \beta^{-}_{m,w}\) if \(\frac{\left(\delta - \delta_{0} \beta^{-}_{m,w}\right) \psi}{\xi} > 0\). Note that, \(\delta - \delta_{0} \beta^{-}_{m,w} > 0\), since \(\delta \geq \delta_{0}\) and \(\beta^{-}_{m,w} < 1\). Thus, all we need to do is to investigate the sign of \(\frac{\psi\left(\overline{w}, \beta^{-}_{m,w}\right)}{\xi\left(\delta, \beta^{-}_{m,w}\right)}\).

Note that \(\xi\left(\delta, \beta^{-}_{m,w}\right)\) and \(\xi\left(0, \beta^{-}_{m,w}\right)\) have the same sign provided that \(\overline{\delta}\) is sufficiently small and \(\xi\left(0, \beta^{-}_{m,w}\right) \neq 0\). Following Lemma \ref{Lem8tvi7tvli7ut6rdkawtudfkd}, \(\xi\left(0, \beta^{-}_{m,w}\right) = -2\sin\theta \tan\frac{n\theta}{2}\), where \(\theta = \pi - \cos^{-1}(\frac{\beta^{-}_{m,w} -2}{2})\). It is easily seen that \(\xi\left(0, \beta^{-}_{m,w}\right) < 0\) if \(0 < \beta^{-}_{m,w} < \alpha_{1}\), and \(\xi\left(0, \beta^{-}_{m,w}\right) > 0\) if \(\alpha_{1} < \beta^{-}_{m,w} < \alpha_{2}\). On the other hand, by Lemma \ref{Lem8tvi7tvli7ut6rdkawtudfkd}, 
\begin{equation*}
\psi\left(\overline{w}, \beta^{-}_{m,w}\right) = \frac{1}{\cos\frac{n\theta}{2}}\sum_{i=0}^{n-1} w_{i}\cos\left(\frac{n}{2} - i\right)\theta,
\end{equation*}
where \(\theta = \pi - \cos^{-1}(\frac{\beta^{-}_{m,w} -2}{2})\). Note that \(\cos\frac{n\theta}{2} < 0\) if \(\alpha_{1} < \beta^{-}_{m,w} < \alpha_{2}\), and \(\cos\frac{n\theta}{2} > 0\) if \(0 < \beta^{-}_{m,w} < \alpha_{1}\). Moreover, when \(0 < \beta^{-}_{m,w} < \alpha_{1}\), we have that \(\cos\left(\frac{n}{2} - i\right)\theta > 0\) for \(i=0,\ldots, n-1\), and therefore \(\sum_{i=0}^{n-1} w_{i}\cos\left(\frac{n}{2} - i\right)\theta > 0\). We have
\begin{equation*}
\left\lbrace\begin{array}{ll}
\frac{\psi\left(\overline{w}, \beta^{-}_{m,w}\right)}{\xi\left(0, \beta^{-}_{m,w}\right)} < 0 \quad & \mathrm{ if }\,\, 0 < \beta^{-}_{m,w} < \alpha_{1},\\
\frac{\psi\left(\overline{w}, \beta^{-}_{m,w}\right)}{\xi\left(0, \beta^{-}_{m,w}\right)} > 0 \quad & \mathrm{ if }\,\, \alpha_{1} < \beta^{-}_{m,w} < \alpha_{2}, \,\,\mathrm{ and }\,\, \sum_{i=0}^{n-1} w_{i}\cos\left(\frac{n}{2} - i\right)\theta < 0,\\
\frac{\psi\left(\overline{w}, \beta^{-}_{m,w}\right)}{\xi\left(0, \beta^{-}_{m,w}\right)} < 0 \quad & \mathrm{ if }\,\, \alpha_{1} < \beta^{-}_{m,w} < \alpha_{2}, \mathrm{\,\, and \,\,} \sum_{i=0}^{n-1} w_{i}\cos\left(\frac{n}{2} - i\right)\theta > 0.
\end{array}\right.
\end{equation*}
This ends the proof of parts (\ref{Item75988tugqyefriyer111}), (\ref{Item75988tugqyefriyer222}) and (\ref{Item75988tugqyefriyer333}) of Theorem \ref{Thm274i7tigfdyf}.
\end{proof}

\begin{proof}[Proof of part (\ref{Item89a8grye8gfaergeww}) of Theorem \ref{Thm274i7tigfdyf}]
Since \(\alpha_{2} < \beta^{-}_{m, w}\), we have \(\lambda_{2}\left(L_{G}\right) = \alpha_{2}\) (see Proposition \ref{Prop87986kiuglru3aw3er}). On the other hand, by Lemma \ref{Lem50hugiughbaujgrgadf}, \(S\) and \(p^{\prime}\left(\frac{2\pi}{n}\right)\) have the same sign. If \(S < 0\), since \(p\left(\frac{\pi}{n}\right) < 0\) (see Lemma \ref{Lem50hugiughbaujgrgadf}), the intermediate value theorem implies that \(L_{G_p}\) has an eigenvalue (\(p\) has a root) smaller than \(\alpha_{2}\). Denote this eigenvalue by \(\lambda\left(\overline{\delta}\right)\). When the modification is small, this eigenvalue is indeed the spectral gap of \(L_{G_{p}}\), as discussed in the proof of part (\ref{Item42iyo8yaor4gtwf1}). However, for large modification, there is the possibility of the emergence of non-real eigenvalues of \(L_{G_{p}}\). In such a scenario, there might be complex conjugates eigenvalues of \(L_{G_p}\) whose real part decreases and becomes smaller than \(\lambda\left(\overline{\delta}\right)\). This means that the spectral gap is not necessarily a real number, however, since \(\lambda\left(\overline{\delta}\right) < \alpha_2\), we always have \(\mathrm{Re}\left(\lambda_{2}\left(L_{G_{p}}\right)\right) < \lambda_{2}\left(L_{G}\right)\). This proves part (\ref{Item538youkugaliyfyejr11}) of the theorem.

The proof of part (\ref{Item538youkugaliyfyejr11}) is similar. For small modifications, as discussed in the proof of part (\ref{Item42iyo8yaor4gtwf2}) of the theorem, we have that \(\lambda_{2}\left(L_{G_{p}}\right) = \alpha_{2}\). In fact, \(\alpha_{2}\) is an eigenvalue of \(L_{G_p}\) for arbitrary \(\overline{\delta}\). However, as we discussed above, there is a possibility of the emergence of non-real eigenvalues for \(L_{G_{p}}\) when the modification \(\overline{\delta}\) is large. Thus, this might be the case that the real parts of these non-real eigenvalues reduce, and they become the spectral gap of \(L_{G_p}\). In any case, the property \(\mathrm{Re}\left(\lambda_{2}\left(L_{G_{p}}\right)\right) \leq \lambda_{2}\left(L_{G}\right)\) always holds. This proves part (\ref{Item538youkugaliyfyejr22}) of the theorem.
\end{proof}

\section{Conclusions}

This paper investigates how modifying a network affects the Laplacian spectral gap and, in turn, the collective dynamics (synchronization). We considered modifications of three networks with master-slave topology, where the master is a cycle, and the slave is a star. The cycle and star were chosen since they are common motifs in real-world networks. The considered modifications are of arbitrary size and not necessarily (sufficiently) small. Our investigation was based on the spectral analysis of the Laplacian matrices of these networks. Our results are rigorous and accompanied by simulations of networks of coupled Lorenz oscillators that admit our mathematical results.

One particular interest of this paper was a paradoxical scenario known as Braess's paradox in which improving network connectivity leads to functional failure, such as synchronization loss. We explored and classified such scenarios in these three network models. We have shown that this counter-intuitive scenario in our models is not rare at all. For instance, a critical value exists (proportional to the square root of the size of the star) for which Braess's paradox happens in models I and II if the cycle size exceeds this value.

\section*{Acknowledgements}

SB, NK, and DE were supported by TUBITAK Grant No. 119F125. DE was supported by the BAGEP Award of the Science Academy, Turkey. TP was supported by a Newton Advanced Fellowship of the Royal Society NAF\textbackslash R1\textbackslash 180236, by Serrapilheira Institute (Grant No. Serra- 1709-16124), and FAPESP (grant 2013/07375-0). SB and TP acknowledge the support of FAPESP (grant no. 2023/04294-0). SB acknowledges the support of Leverhulme Trust grant RPG-2021-072.

\section*{Data Availability}

The datasets generated during and/or analyzed during the current study are available from the authors at reasonable request.

\section*{Declarations}

\textbf{Conflict of interest} The authors declare that there is no conflict of interest.

\setcounter{theorem}{0}
\renewcommand{\thetheorem}{\Alph{section}.\arabic{theorem}} % I added these commands to avoid numberings like "Lemma Appendix A.1". By this command, we get "Lemma A.1."

\appendix

\section{Synchronization of coupled R{\"o}ssler oscillators}\label{syncRoss}

We consider the network given model II, with an isolated isolated dynamics $f:\mathbb{R}^3 \rightarrow \mathbb{R}^3$ as the R{\"o}ssler oscillator, described as
\begin{equation}
\begin{split}
    \dot{x}&=-y-z\\
    \dot{y}&=x+0.2y\\
    \dot{z}&=0.2+z(x-5.7).
    \end{split}
\end{equation}
Its dynamics is also chaotic like the Lorenz attractor.  

\subsection{Hindering synchronization}

At time zero, we consider model II, where the sizes of the cycle and star subgraphs are  $n = m = 15$. They are initially connected via a directed link from the cycle to the star subgraph where \(w_0 = 1\), as shown in the left upper panel of Figure \ref{FigRos}.

We consider $H=\mathbf{I}$ as the coupling function between oscillators. We randomly choose initial conditions from the uniform distribution over \([0.5, 1)\), and integrate the network until time $t=1000s$. In this master-slave configuration, the network converges towards a synchronous motion, where the mean error $\langle E\rangle$ goes to zero.
At time $t=1000s$, we break the master-slave configuration by adding the cutset edges \(\delta_i = 1\), where \(i = 0,1,\hdots,m - 1\). This is represented as the red edges on the right upper panel of Figure \ref{FigRos}. Along with the network modification, we introduce small perturbations to each state randomly selected from the uniform distribution over \([0.01, 0.02)\). As observed, this modification leads to an instability of the synchronous motion. 

We can use Theorem B to understand the hindrance of synchronization. We add the cutset edges \(\delta_i = 1\), where \(i = 0,1,\hdots,m - 1\), which corresponds to the global modification case in Theorem B. We focus to the case \(b\) since \(\alpha_2=0.17, \beta_{15,1}^-=0.06\). We observe that \(\delta<\delta_0 \beta_{15,1}^+\) where \(\delta=15,\) \(\delta_0=1\), \(\beta_{15,1}^+=15.94\). Thus, we can use local modification results from Theorem B. The theorem says \(\lambda_2(L_{G_p})<\lambda_2(L_G)\) since \(\alpha_1<\beta_{15,1}^-\) where \(\alpha_1=0.04\). In this case, we know that \(G\) is more synchronizable than \(G_p\) by Definition 1.

\begin{figure}[h]
\includegraphics[width=16cm]{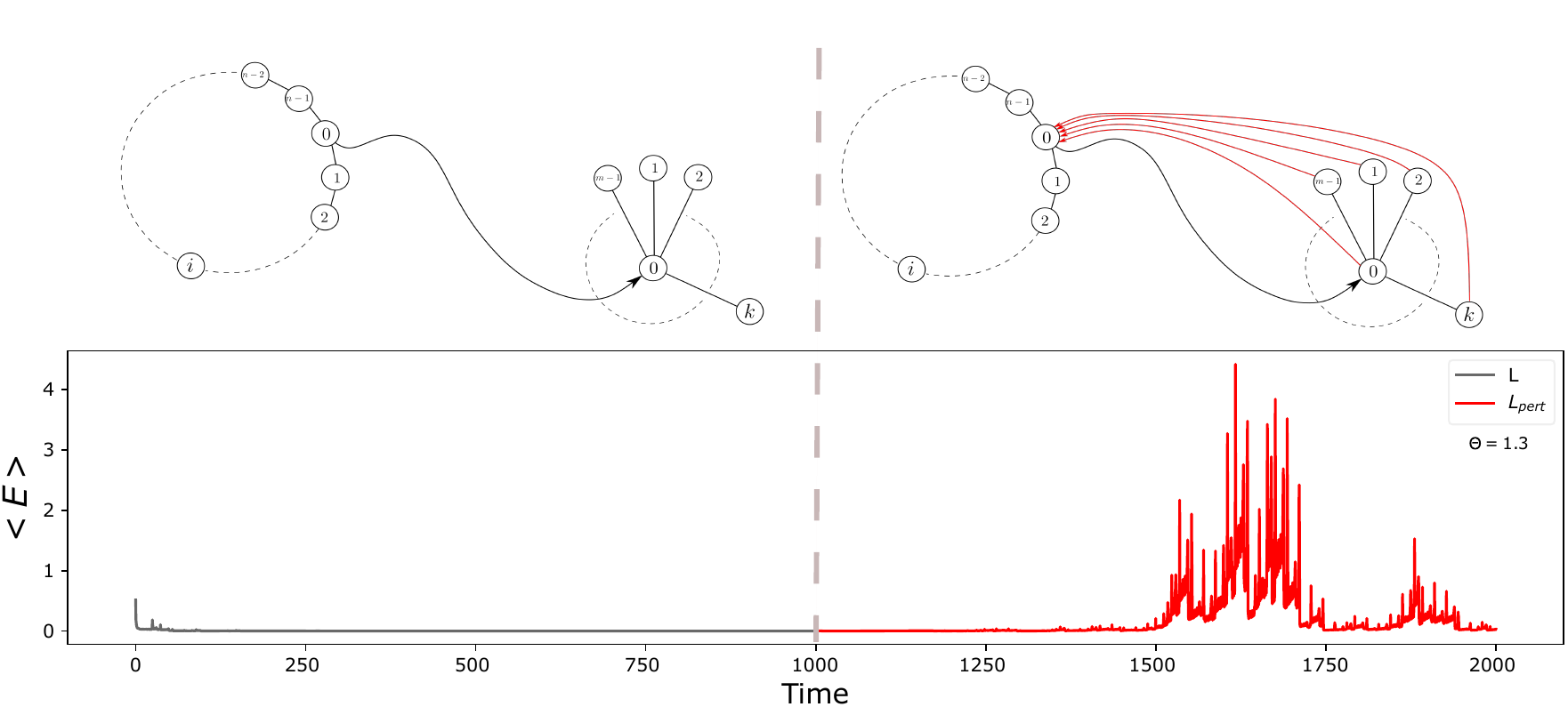} \caption{ Hindrance of synchronization due to link addition in networks of coupled R{\"o}ssler oscillators. The figure shows the synchronization error over time with link addition occurring at time $t=1000s$. The sizes of the cycle and star subgraphs are set to $n = 15$ and $m = 15$, and subgraphs are connected via a directed link from the cycle subgraph to the star subgraph where \(w_0 = 1\). At $t=1000s$, we add the red links to each system with unit weight. After perturbation, the system doesn't return to synchronization.} \label{FigRos}
\end{figure}

\subsection{Enhancing synchronization}

Again, we consider model II with $n = 10$ and $m = 20$ along with  $H=\mathbf{I}$ as the coupling function as before. 

%All initial states are randomly selected from the uniform distribution over \([0.5, 1)\).

We randomly choose initial conditions from the uniform distribution over \([0.5, 1)\), and integrate the network until time $t=1000s$. In this master-slave configuration, the network does not synchronize, as can be observed in the the mean synchronization error $\langle E\rangle$. At time $t=1000s$, we break the master-slave configuration by adding a cutset and modification edges are \(w_0 = 1\) and \(\delta_i = 1\), where \(i = 0,1,\hdots,m - 1\). This is represented as the red edges on the right upper panel of Figure \ref{FigRos_enhance}. After this modification, the synchronous state becomes stable, and the synchronization error converges to zero. 

\begin{figure}[h]
\includegraphics[width=16cm]{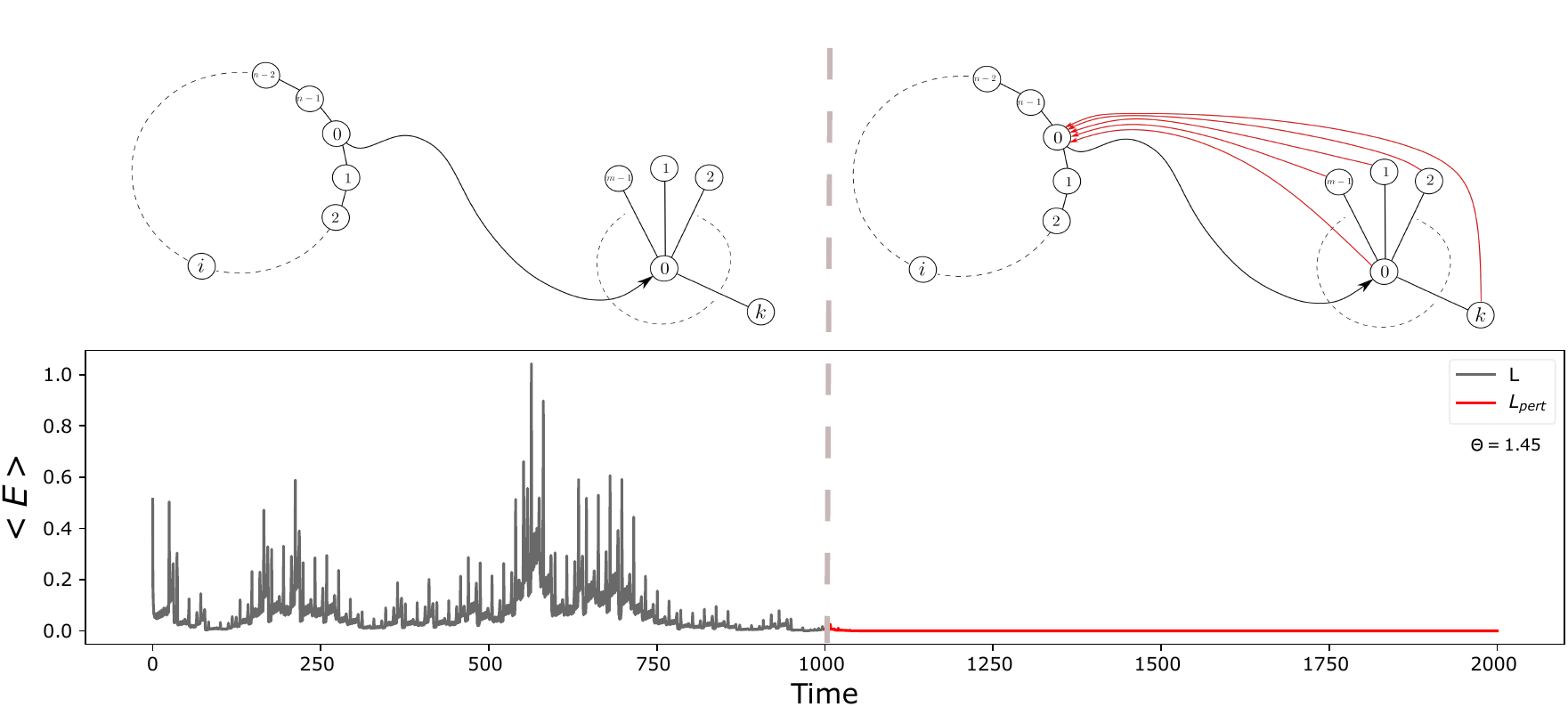}
 \caption{Enhancement of synchronization due to link addition in a network of coupled R{\"o}ssler oscillators.  The sizes of the cycle and star subgraphs are  $n = 10$ and $m = 20$, initially, and subgraphs are connected via a directed link from the cycle to the star where \(w_0 = 1\). We integrate the network until time $t=1000s$ and find the synchronization is unstable.  After the red links are added to the system at $t=1000s$,  the synchronization mean error $\langle E\rangle$ goes to zero, indicating that the modification leads to stable synchronous dynamics.} \label{FigRos_enhance}
\end{figure}

Again, we use Theorem B to deduce the synchronization enhancement. We add the cutset edges \(\delta_i = 1\), where \(i = 0,1,\hdots,m - 1\), which corresponds to the global modification case in Theorem B. We focus to the case \(b\) since \(\alpha_2=0.38, \beta_{20,1}^-=0.05\). We observe that \(\delta<\delta_0 \beta_{20,1}^+\) where \(\delta=20,\) \(\delta_0=1\), \(\beta_{20,1}^+=20.95\). Thus, we can use local modification results. The theorem says \(\lambda_2(L_{G})<\lambda_2(L_{G_p})\) since \(\beta_{20,1}^-<\alpha_1\) where \(\alpha_1=0.1\). In this case, we know that \(G_p\) is more synchronizable than \(G\) by Definition 1.

\section{Technical lemmas}\label{Section4311ohbiwubwavs}

\begin{lemma}\label{lem11y78ty998796f56drytr}
Consider the block matrix \(\left(\begin{smallmatrix}
A & B\\
C & D
\end{smallmatrix}\right)\) and assume \(D\) is invertible. Define \(E: = A - B D^{-1} C\). Then
\begin{enumerate}[(i)]
\item \(\mathrm{det}\left(\begin{smallmatrix}
A & B\\
C & D
\end{smallmatrix}\right) = \mathrm{det}\left(D\right) \times \mathrm{det}\left(E\right)\).
\item Assume \(E\) is invertible. Then, \(\left(\begin{smallmatrix}
A & B\\
C & D
\end{smallmatrix}\right)^{-1} = \left(\begin{smallmatrix}
E^{-1} & - E^{-1} BD^{-1}\\
-D^{-1} C E^{-1} & D^{-1} + D^{-1}C {E}^{-1} BD^{-1}
\end{smallmatrix}\right)\).
\end{enumerate}
\end{lemma}
\begin{proof}
See, e.g., \cite{meyer2000matrix}.
\end{proof}

\begin{lemma}\label{Lem391q0kieb6fb2j2iwbjyfu5ewwe}
For \(m\geq 4\), let \(X_{0}\), \(X_{1}\),..., \(X_{m-1}\) be real-valued functions defined on a subset of \(\mathbb{R}\), and consider the \(m\times m\) matrix
\begin{equation*}
M = M\left(\lambda\right) = \left(\begin{array}{c|c}
m -\lambda  + X_{0}\left(\lambda\right) &
\begin{array}{cccc}
-1 + X_{1}\left(\lambda\right) & -1 + X_{2}\left(\lambda\right) & \cdots & -1 + X_{m-1}\left(\lambda\right)
\end{array}
\\
\hline
\mathbf{1}_{m-1} & \left(1-\lambda\right) \mathrm{I}
\end{array}\right).
\end{equation*}
Let \(\lambda_0\in\mathbb{R}\), and assume that all the functions \(X_{i}\) are defined at \(\lambda_{0}\). 
\begin{enumerate}[(i)]
\item\label{Itemm892f34rgufgrye9uyd111} There exists \(0\neq v\in \mathbb{R}^m\) such that \(M(\lambda_{0}) v = 0\) if and only if \(\lambda_0 = 1\) or \(\lambda = \lambda_0\) satisfies
\begin{equation}\label{eq52kugyuyfgvhgcutcyrsaalow}
\lambda^{2} - \left[1 + m + X_{0}\left(\lambda\right)\right]\lambda + 1 + \sum_{i=0}^{m-1} X_{i}\left(\lambda\right) = 0.
\end{equation}
\item\label{Itemm892f34rgufgrye9uyd222} Assume that there exists \(1\leq i\leq m-1\) such that \(X_{i}(1) \neq 1\). Then, the vector subspace \(E^{\mathrm{right}}\subset \mathbb{R}^{m}\) (resp. \(E^{\mathrm{left}}\subset \mathbb{R}^{m}\)) of all the solutions \(v\) of the equation \(M(1) v = 0\) (resp. \(v^{\top} M(1) = 0\)) has \(m-2\) dimensions.
\item\label{Itemm892f34rgufgrye9uyd333} Let \(v = (v_0, v_1, \ldots, v_{m-1})\in\mathbb{R}^{m}\). Then, any \(v \in E^{\mathrm{right}}\) satisfies the property \(v_{0} = 0\). Moreover, if there exists \(1\leq i\leq m-1\) such that \(X_{i}(1) \neq 1\), then, any \(v \in E^{\mathrm{left}}\) satisfies the property \(v_{0} = 0\) too.
\end{enumerate}
\end{lemma}

\begin{proof}
For \(v = (v_0, v_1, \ldots, v_{m-1})\in\mathbb{R}^{m}\), we have \(M(\lambda) v = 0\) if and only if
\begin{eqnarray}
& \left[m-\lambda + X_{0}\left(\lambda\right)\right] v_{0} + \sum_{i=1}^{m-1} \left[-1 + X_{i}\left(\lambda\right)\right] v_{i} = 0,\quad \mathrm{and}\label{eq02y7tyfvnaokjeyewrwyqqqa}\\
& v_{0} + \left(1-\lambda\right) v_{i} = 0,\quad \mathrm{for\,\, all \,\,} 1 \leq i \leq m-1.\label{eq909798it78r6futftdtejbzjhvs}
\end{eqnarray}
This implies that for \(\lambda = 1\), the equation \(M(1) v = 0\) has a solution \(v\) if and only if \(v\in E^{\mathrm{right}}\) for
\begin{equation}\label{eq80279867tuyfwytchgcss}
E^{\mathrm{right}} :=\lbrace v\in\mathbb{R}^{m} \vert\quad v_{0} = 0 \mathrm{\,\, and \,\,} \langle\left(v_{1}, v_{2}, \cdots, v_{m-1}\right), \left(X_{1}\left(1\right) - 1,\ldots, X_{m-1}\left(1\right)-1\right)\rangle = 0\rbrace,
\end{equation}
where \(\langle \cdot , \cdot \rangle\) is the standard Euclidean inner product in \(\mathbb{R}^{m-1}\).

If \(M(1)v = 0\), for \(\lambda\neq 1\), has a solution \(v\), then it follows from  (\ref{eq909798it78r6futftdtejbzjhvs}) that \(v_{0} \neq 0\) and \(v_{i} = \frac{v_{0}}{1-\lambda}\), for all \(1 \leq i \leq m-1\). Substituting this into (\ref{eq02y7tyfvnaokjeyewrwyqqqa}), and multiplying the derived equation by \(1-\lambda\) give (\ref{eq52kugyuyfgvhgcutcyrsaalow}). This proves part (\ref{Itemm892f34rgufgrye9uyd111}).

Part (\ref{Itemm892f34rgufgrye9uyd222}) follows from (\ref{eq80279867tuyfwytchgcss}). To prove part (\ref{Itemm892f34rgufgrye9uyd333}), let \(v\in E^{\mathrm{left}}\). By \(v^{\top} M(1) = 0\), we have that if there exists \(1\leq i \leq m-1\) such that \(X_{i}(1)\neq 1\), then \(v_{0} = 0\) and \(v_{1} + v_{2} + \cdots + v_{m-1} = 0\). The set of all \(v\)s satisfying these properties is a \((m-2)\)-dimensional subspace of \(\mathbb{R}^{m}\). This finishes the proof of the lemma.
\end{proof}

\bibliography{Myreferences.bib}

\begin{thebibliography}{36}
\providecommand{\natexlab}[1]{#1}
\providecommand{\url}[1]{\texttt{#1}}
\expandafter\ifx\csname urlstyle\endcsname\relax
  \providecommand{\doi}[1]{doi: #1}\else
  \providecommand{\doi}{doi: \begingroup \urlstyle{rm}\Url}\fi

\bibitem[Ermentrout and Terman(2010)]{ermentrout2010mathematical}
B.~Ermentrout and D.~H. Terman.
\newblock \emph{Mathematical foundations of neuroscience}, volume~35.
\newblock Springer, New York, 2010.

\bibitem[Newman(2018)]{newman2018networks}
M.~Newman.
\newblock \emph{Networks}.
\newblock Oxford university press, Oxford, 2018.

\bibitem[Eroglu et~al.(2017)Eroglu, Lamb, and Pereira]{Eroglu2017synchronisation}
D.~Eroglu, J.~S.~W. Lamb, and T.~Pereira.
\newblock Synchronisation of chaos and its applications.
\newblock \emph{Contemporary Physics}, 58\penalty0 (3):\penalty0 207--243, 2017.

\bibitem[Prasad et~al.(2010)Prasad, Dhamala, Adhikari, and Ramaswamy]{Prasad2010amplitude}
A.~Prasad, M.~Dhamala, B.~M. Adhikari, and R.~Ramaswamy.
\newblock Amplitude death in nonlinear oscillators with nonlinear coupling.
\newblock \emph{Physical Review E}, 81\penalty0 (2):\penalty0 027201, 2010.

\bibitem[Louodop et~al.(2019)Louodop, Tchitnga, Fagundes, Kountchou, Tamba, Pando, and Cerdeira]{Louodop2019extreme}
P.~Louodop, R.~Tchitnga, F.~F. Fagundes, M.~Kountchou, V.~K. Tamba, C.~L. Pando, and H.~A. Cerdeira.
\newblock Extreme multistability in a josephson-junction-based circuit.
\newblock \emph{Physical Review E}, 99\penalty0 (4):\penalty0 042208, 2019.

\bibitem[Aguiar et~al.(2019)Aguiar, Dias, and Field]{Aguiar2019feedforward}
M.~A.~D. Aguiar, A.~Dias, and M.~Field.
\newblock Feedforward networks: adaptation, feedback, and synchrony.
\newblock \emph{Journal of Nonlinear Science}, 29\penalty0 (3):\penalty0 1129--1164, 2019.

\bibitem[Field(2015)]{Field2015heteroclinic}
M.~Field.
\newblock Heteroclinic networks in homogeneous and heterogeneous identical cell systems.
\newblock \emph{Journal of Nonlinear Science}, 25\penalty0 (3):\penalty0 779--813, 2015.

\bibitem[Eldan et~al.(2017)Eldan, R{\'a}cz, and Schramm]{Eldan2017braess}
R.~Eldan, M.~Z. R{\'a}cz, and T.~Schramm.
\newblock Braess's paradox for the spectral gap in random graphs and delocalization of eigenvectors.
\newblock \emph{Random Structures \& Algorithms}, 50\penalty0 (4):\penalty0 584--611, 2017.

\bibitem[Pade and Pereira(2015)]{Pade2015improving}
J.~P. Pade and T.~Pereira.
\newblock Improving network structure can lead to functional failures.
\newblock \emph{Scientific reports}, 5\penalty0 (1):\penalty0 1--6, 2015.

\bibitem[Nishikawa and Motter(2010)]{nishikawa2010network}
T.~Nishikawa and A.~E. Motter.
\newblock Network synchronization landscape reveals compensatory structures, quantization, and the positive effect of negative interactions.
\newblock \emph{Proceedings of the National Academy of Sciences}, 107\penalty0 (23):\penalty0 10342--10347, 2010.

\bibitem[Pereira et~al.(2020)Pereira, Van~Strien, and Tanzi]{TanziJEMS2020}
T.~Pereira, S.~Van~Strien, and M.~Tanzi.
\newblock Heterogeneously coupled maps: hub dynamics and emergence across connectivity layers.
\newblock \emph{Journal of the European Mathematical Society}, 22\penalty0 (7):\penalty0 2183--2252, 2020.

\bibitem[Eguiluz et~al.(2005)Eguiluz, Chialvo, Cecchi, Baliki, and Apkarian]{eguiluz2005scale}
V.~M. Eguiluz, D.~R. Chialvo, G.~A. Cecchi, M.~Baliki, and A.~V.~S. Apkarian.
\newblock Scale-free brain functional networks.
\newblock \emph{Physical review letters}, 94\penalty0 (1):\penalty0 018102, 2005.

\bibitem[Chung(1997)]{Chung1997spectral}
Fan~RK Chung.
\newblock \emph{Spectral graph theory}, volume~92.
\newblock American Mathematical Soc., 1997.

\bibitem[Poignard et~al.(2019)Poignard, Pade, and Pereira]{poignard2019effects}
C.~Poignard, J.~P. Pade, and T.~Pereira.
\newblock The effects of structural perturbations on the synchronizability of diffusive networks.
\newblock \emph{Journal of Nonlinear Science}, 29\penalty0 (5):\penalty0 1919--1942, 2019.

\bibitem[Biyikoglu et~al.(2007)Biyikoglu, Leydold, and Stadler]{biyikoglu2007laplacian}
T{\"u}rker Biyikoglu, Josef Leydold, and Peter~F Stadler.
\newblock \emph{Laplacian eigenvectors of graphs: Perron-Frobenius and Faber-Krahn type theorems}.
\newblock Springer, 2007.

\bibitem[Hart et~al.(2015)Hart, Pade, Pereira, Murphy, and Roy]{Hart2015adding}
J.~D. Hart, J.~P. Pade, T.~Pereira, T.~E. Murphy, and R.~Roy.
\newblock Adding connections can hinder network synchronization of time-delayed oscillators.
\newblock \emph{Physical Review E}, 92\penalty0 (2):\penalty0 022804, 2015.

\bibitem[Milanese et~al.(2010)Milanese, Sun, and Nishikawa]{Milanese2010approximating}
A.~Milanese, J.~Sun, and T.~Nishikawa.
\newblock Approximating spectral impact of structural perturbations in large networks.
\newblock \emph{Physical Review E}, 81\penalty0 (4):\penalty0 046112, 2010.

\bibitem[Ma'ayan et~al.(2008)Ma'ayan, Cecchi, Wagner, Rao, Iyengar, and Stolovitzky]{ma2008ordered}
A.~Ma'ayan, G.~A. Cecchi, J.~Wagner, A.~R. Rao, R.~Iyengar, and G.~Stolovitzky.
\newblock Ordered cyclic motifs contribute to dynamic stability in biological and engineered networks.
\newblock \emph{Proceedings of the National Academy of Sciences}, 105\penalty0 (49):\penalty0 19235--19240, 2008.

\bibitem[Kashtan and Alon(2005)]{kashtan2005spontaneous}
N.~Kashtan and U.~Alon.
\newblock Spontaneous evolution of modularity and network motifs.
\newblock \emph{Proceedings of the National Academy of Sciences}, 102\penalty0 (39):\penalty0 13773--13778, 2005.

\bibitem[Alexander et~al.(1986)Alexander, DeLong, and Strick]{alexander1986}
Garrett~E Alexander, Mahlon~R DeLong, and Peter~L Strick.
\newblock Parallel organization of functionally segregated circuits linking basal ganglia and cortex.
\newblock \emph{Annual review of neuroscience}, 9\penalty0 (1):\penalty0 357--381, 1986.

\bibitem[Ben-Yishai et~al.(1997)Ben-Yishai, Hansel, and Sompolinsky]{ben1997traveling}
Rani Ben-Yishai, David Hansel, and Haim Sompolinsky.
\newblock Traveling waves and the processing of weakly tuned inputs in a cortical network module.
\newblock \emph{Journal of computational neuroscience}, 4:\penalty0 57--77, 1997.

\bibitem[Bonifazi et~al.(2009)Bonifazi, Goldin, Picardo, Jorquera, Cattani, Bianconi, Represa, Ben-Ari, and Cossart]{bonifazi2009gabaergic}
Paolo Bonifazi, Miri Goldin, Michel~A Picardo, Isabel Jorquera, A~Cattani, Gregory Bianconi, Alfonso Represa, Yehezkel Ben-Ari, and Rosa Cossart.
\newblock Gabaergic hub neurons orchestrate synchrony in developing hippocampal networks.
\newblock \emph{Science}, 326\penalty0 (5958):\penalty0 1419--1424, 2009.

\bibitem[Vlasov et~al.(2015)Vlasov, Zou, and Pereira]{vlasov2015explosive}
Vladimir Vlasov, Yong Zou, and Tiago Pereira.
\newblock Explosive synchronization is discontinuous.
\newblock \emph{Physical Review E}, 92\penalty0 (1):\penalty0 012904, 2015.

\bibitem[T{\"o}njes et~al.(2021)T{\"o}njes, Fiore, and Pereira]{tonjes2021coherence}
Ralf T{\"o}njes, Carlos~E Fiore, and Tiago Pereira.
\newblock Coherence resonance in influencer networks.
\newblock \emph{Nature Communications}, 12\penalty0 (1):\penalty0 72, 2021.

\bibitem[Mersing et~al.(2021)Mersing, Tyler, Ponboonjaroenchai, Tinsley, and Showalter]{mersing2021novel}
David Mersing, Shannyn~A Tyler, Benjamas Ponboonjaroenchai, Mark~R Tinsley, and Kenneth Showalter.
\newblock Novel modes of synchronization in star networks of coupled chemical oscillators.
\newblock \emph{Chaos: An Interdisciplinary Journal of Nonlinear Science}, 31\penalty0 (9):\penalty0 093127, 2021.

\bibitem[Muni and Provata(2020)]{muni2020chimera}
S.~S. Muni and A.~Provata.
\newblock Chimera states in ring--star network of chua circuits.
\newblock \emph{Nonlinear Dynamics}, 101\penalty0 (4):\penalty0 2509--2521, 2020.

\bibitem[Kantner and Yanchuk(2013)]{kantner2013bifurcation}
Markus Kantner and Serhiy Yanchuk.
\newblock Bifurcation analysis of delay-induced patterns in a ring of hodgkin--huxley neurons.
\newblock \emph{Philosophical Transactions of the Royal Society A: Mathematical, Physical and Engineering Sciences}, 371\penalty0 (1999):\penalty0 20120470, 2013.

\bibitem[Manik et~al.(2017)Manik, Timme, and Witthaut]{manik2017cycle}
Debsankha Manik, Marc Timme, and Dirk Witthaut.
\newblock Cycle flows and multistability in oscillatory networks.
\newblock \emph{Chaos: An Interdisciplinary Journal of Nonlinear Science}, 27\penalty0 (8):\penalty0 083123, 2017.

\bibitem[Corder et~al.(2023)Corder, Bian, Pereira, and Montalb{\'a}n]{corder2023emergence}
Rodrigo~M Corder, Zheng Bian, Tiago Pereira, and Antonio Montalb{\'a}n.
\newblock Emergence of chaotic cluster synchronization in heterogeneous networks.
\newblock \emph{Chaos: An Interdisciplinary Journal of Nonlinear Science}, 33\penalty0 (9), 2023.

\bibitem[Brouwer and Haemers(2011)]{Brouwer2011spectra}
A.~E. Brouwer and W.~H. Haemers.
\newblock \emph{Spectra of graphs}.
\newblock Springer Science \& Business Media, New York, 2011.

\bibitem[Mohar(1997)]{mohar1997some}
Bojan Mohar.
\newblock Some applications of laplace eigenvalues of graphs.
\newblock In \emph{Graph symmetry: Algebraic methods and applications}, pages 225--275. Springer, 1997.

\bibitem[Horn and Johnson(2012)]{horn2012matrix}
R.~Horn and C.~R. Johnson.
\newblock \emph{Matrix analysis}.
\newblock Cambridge university press, New York, 2012.

\bibitem[Pereira et~al.(2014)Pereira, Eldering, Rasmussen, and Veneziani]{Pereira2014towards}
T.~Pereira, J.~Eldering, M.~Rasmussen, and A.~Veneziani.
\newblock Towards a theory for diffusive coupling functions allowing persistent synchronization.
\newblock \emph{Nonlinearity}, 27\penalty0 (3):\penalty0 501--525, 2014.

\bibitem[Veerman and Lyons(2020)]{Veerman2020}
J.~J.~P. Veerman and R.~Lyons.
\newblock A primer on laplacian dynamics in directed graphs.
\newblock \emph{Nonlinear Phenomena in Complex Systems}, 23:\penalty0 196--206, 2020.

\bibitem[Hu and O'Connell(1996)]{hu1996analytical}
G.~Y. Hu and R.~F. O'Connell.
\newblock Analytical inversion of symmetric tridiagonal matrices.
\newblock \emph{Journal of Physics A: Mathematical and General}, 29\penalty0 (7):\penalty0 1511, 1996.

\bibitem[Meyer(2000)]{meyer2000matrix}
C.~D. Meyer.
\newblock \emph{Matrix analysis and applied linear algebra}, volume~71.
\newblock SIAM, Philadelphia, 2000.

\end{thebibliography}

\end{document}